\documentclass[times, 10pt,a4wide,envcountsame]{llncs}

\usepackage{a4wide}
\usepackage{times}
\usepackage{amsmath,amssymb,eufrak,theorem}
\usepackage{enumerate}
\usepackage{complexity}
\usepackage{gastex}

\usepackage{stmaryrd,rotating,mathrsfs}

\usepackage[active]{srcltx}

\def\AC{{\mathsf{AC}}}

\def\AND{{\mathrm{AND}}}
\def\BIN{{\mathrm{BIN}}}
\def\CRR{{\mathrm{CRR}}}
\def\CTL{{\mathsf{CTL}}}
\def\div{{\mathrm{div}}}

\def\EF{{\mathsf{EF}}}
\def\EXPTIME{{\mathsf{EXPTIME}}}
\def\inp{{\mathrm{in}}}
\def\leaf{{\mathrm{leaf}}}

\def\NC{{\mathsf{NC}}}
\def\Omc{{\mathcal{O}}}
\def\OR{{\mathrm{OR}}}
\def\out{{\mathrm{out}}}
\def\Pmc{{\mathcal{P}}}
\def\path{{\mathrm{path}}}
\def\PSPACE{{\mathsf{PSPACE}}}
\def\TC{{\mathsf{TC}}}
\def\U{{\mathsf{U}}}

\newcommand{\ValOne}{\text{ValOne}}
\newcommand{\OptValOne}{\text{OptValOne}}
\newcommand{\LCM}{\text{LCM}}

\newtheorem{fact}[theorem]{Fact}

\begin{document}

\title{Branching-time model checking of one-counter processes}
\author{Stefan G\"oller\inst{1} \and Markus Lohrey\inst{2,}\thanks{The second author
would like to acknowledge the support by DFG research project GELO.}}
\institute{
Universit\"at
Bremen, Fachbereich Mathematik und Informatik, Germany
\and
Universit\"at
Leipzig, Institut f\"ur Informatik, Germany\\
\email{goeller@informatik.uni-bremen.de
lohrey@informatik.uni-leipzig.de}}

\maketitle

\begin{abstract}
One-counter processes (OCPs) are pushdown processes which operate only
on a unary stack alphabet. We study the computational complexity of
model checking
computation tree logic ($\CTL$) over OCPs. 
A $\PSPACE$ upper bound is inherited from the modal $\mu$-calculus 
for this problem \cite{Serr06}.
First, we analyze the periodic behaviour of $\CTL$ over OCPs and
derive a model checking algorithm whose running time is exponential
only in the number of control locations and a syntactic notion of the 
formula that we call leftward until depth. In particular,
model checking fixed OCPs against $\CTL$ formulas with a fixed 
leftward until depth is in $\P$. This generalizes a corresponding
result from \cite{GoMaTo09} for the expression complexity of 
$\CTL$'s fragment $\EF$. Second, we prove that already over some fixed
OCP, $\CTL$ model checking is $\PSPACE$-hard, i.e., expression
complexity is $\PSPACE$-hard. 
Third, we show that there already exists a fixed $\CTL$ formula for 
which model checking of OCPs is $\PSPACE$-hard, i.e.,
data complexity is $\PSPACE$-hard as well. To obtain the
latter result, we employ two results from complexity theory:
(i) Converting a natural number in Chinese remainder presentation into binary
presentation is in logspace-uniform $\NC^1$ \cite{ChDaLi01} and 
(ii) $\PSPACE$ is $\AC^0$-serializable \cite{HLSVW93}.
We demonstrate that our approach can be used to obtain further results.
We show that model-checking $\CTL$'s fragment $\EF$ over OCPs is hard for $\P^\NP$, thus
establishing a matching lower bound and answering an open question
from \cite{GoMaTo09}.
We moreover show that the following problem is hard for 
$\PSPACE$: Given a one-counter Markov decision process, 
a set of target states with counter value zero each, and an initial
state, to decide whether the probability that the initial state will 
eventually reach one of the target states is arbitrarily close to $1$.
This improves a previously known lower bound for every
level of the Boolean hierarchy shown in \cite{BraBroEteKucWoj09}.
\end{abstract}

\section{Introduction}

Pushdown automata (PDAs) (or recursive state machines; RSMs) are a natural model
for sequential programs with recursive procedure calls, and their verification
problems have been studied extensively.
The complexity of model checking problems for PDAs is quite well understood: 
The reachability problem for PDAs can be solved in  polynomial time
\cite{BoEsMa97,EsHaRoSch00}. Model checking
modal $\mu$-calculus over PDAs was shown to be $\EXPTIME$-complete in
\cite{Wal01}, and the global version of the model checking
problem has been considered in \cite{Cachat02,PitVar04}.
The $\EXPTIME$ lower bound
for model checking PDAs also holds for the simpler logic $\CTL$
and its fragment EG \cite{WalCTL00}, even for a fixed formula
(data complexity) or a fixed PDA (expression complexity).
On the other hand, model checking PDAs against the logic $\EF$
(another natural fragment of $\CTL$) is $\PSPACE$-complete
\cite{WalCTL00}, and again the lower bound still
holds if either the formula or the PDA is fixed \cite{BoEsMa97}.
Model checking problems for various fragments and extensions
of PDL (propositional dynamic logic) over PDAs were studied
in \cite{GolLoh06}.

One-counter processes (OCPs) are Minsky counter machines with just one counter and
action labels on the transitions. They can also be seen as a special
case of PDAs with just one stack symbol, plus a non-removable
bottom symbol which indicates an empty stack (and thus allows to
test the counter for zero) and hence constitute a natural and fundamental
computational model. In recent years, 
model checking problems for OCPs received increasing attention
\cite{GoMaTo09,HKOW-09concur,To09,Serr06}.
Clearly, all upper complexity bounds carry over from PDAs.
The question, whether these upper bounds can be matched by lower bounds
was just recently solved for several important logics:
Model checking $\mu$-calculus over OCPs is $\PSPACE$-complete. 
The $\PSPACE$ upper bound was
shown in \cite{Serr06}, and a matching lower bound can
easily be shown by a reduction from emptiness of alternating
unary finite automata, which was shown to be $\PSPACE$-complete in
\cite{Ho96,JaSa07}.
This lower bound even holds if either the OCP or the
formula is fixed.  The situation becomes different for the fragment $\EF$.
In \cite{GoMaTo09}, it was shown that  model checking $\EF$ over 
OCPs is in the complexity class $\mathsf{P}^{\mathsf{NP}}$ 
(the class of all problems that can be solved on a deterministic
polynomial time machine with access to an oracle from $\mathsf{NP}$).
Moreover, if the input formula is represented succinctly
as a dag (directed acyclic graph), then model checking $\EF$ over 
OCPs is also hard for $\mathsf{P}^{\mathsf{NP}}$.
For  the standard (and less succinct) tree representation for formulas,
only hardness for the class $\mathsf{P}^{\mathsf{NP}[\log]}$
(the class of all problems that can be solved on a deterministic
polynomial time machine which is allowed to make $O(\log(n))$ 
many queries to an oracle from $\mathsf{NP}$) was shown in \cite{GoMaTo09}.
In fact, there already exists a fixed $\EF$ formula
such that model checking this formula over a given OCP is hard for
$\mathsf{P}^{\mathsf{NP}[\log]}$, i.e., the data complexity is 
$\mathsf{P}^{\mathsf{NP}[\log]}$-hard.

In this paper we consider the model checking problem for $\CTL$
over OCPs. 
By the known upper bound for the modal $\mu$-calculus
\cite{Serr06} this problem belongs to $\PSPACE$.
First, we analyze the combinatorics of $\CTL$ model checking over OCPs.
More precisely, we analyze the periodic behaviour of the
set of natural numbers that satisfy a given $\CTL$ formula in a given
control location of the OCP (Theorem~\ref{thm-CTL-periodic}).
By making use of Theorem~\ref{thm-CTL-periodic}, we can derive a model checking
algorithm whose running time is exponential
only in the number of control locations and a syntactic measure on 
$\CTL$ formulas that we call leftward until depth
(Theorem~\ref{CTL upper bound_0}).
As a corollary, we obtain that model checking a fixed OCP
against $\CTL$ formulas of fixed leftward until depth lies in $\P$
(Corollary~\ref{CTL upper bound}). This generalizes a recent result
from \cite{GoMaTo09}, where it was shown that the expression complexity of
$\EF$ over OCPs lies in $\P$.
Next, we focus on lower bounds.
We show that model checking $\CTL$ over OCPs is 
$\PSPACE$-complete, even if we fix either 
the OCP (Theorem~\ref{Theo-CTL-expression}) or the $\CTL$ formula (Theorem~\ref{theo ctl data}).
The proofs for Theorem~\ref{Theo-CTL-expression} uses an intriguing reduction from 
QBF. 
We have to construct a fixed OCP for which we can construct for a given unary
encoded number $i$ $\CTL$ formulas
that express, when interpreted over our fixed OCP, whether the current
counter value is divisible by $2^i$ and whether the $i^{\text{th}}$ bit
in the binary representation of the current counter value is $1$, respectively.  
For the proof of Theorem~\ref{theo ctl data}  ($\PSPACE$-hardness of
data complexity for $\CTL$)
we use two techniques from complexity theory, which to our knowledge have not
been applied in the context of verification so far:
\begin{itemize}
\item the existence of small
depth circuits for converting a number from Chinese remainder representation
to binary representation (see Section~\ref{Sec-circuits} for details) 
and
\item the fact that $\PSPACE$-computations are serializable in a certain sense
(see Section~\ref{sec real} for details).
\end{itemize}
One of the main obstructions 
in getting lower bounds for OCPs is the fact that OCPs are well suited
for testing divisibility properties of the counter value and hence can
deal with numbers in Chinese remainder representation, but it is not 
clear how to deal with numbers in binary representation. Small depth circuits
for converting a number from Chinese remainder representation
to binary representation are the key in order to overcome this obstruction.

We are confident that our new lower bound techniques described above can be used for
proving further lower bounds for OCPs. We present two other applications
of our techniques:
\begin{itemize}
\item We show that model checking $\EF$ over OCPs is complete for $\mathsf{P}^{\mathsf{NP}}$ even if the input formula
is represented by a tree (Theorem~\ref{T EF}) and thereby solve an open problem from \cite{GoMaTo09}.
Figure~\ref{F Model checking} summarizes the picture on the complexity of model
checking for PDAs and OCPs. 
\item  We improve a lower bound on a decision problem for one-counter Markov decision processes from
\cite{BraBroEteKucWoj09} (Theorem~\ref{T ValOne}). More details on this problem
are provided below. 
\end{itemize}
\begin{table}[t]
\begin{center}
\begin{tabular}{l|l|l}
\textbf{Logic}        &  \textbf{PDA}     &  \textbf{OCP}\\ \hline
$\mu$-calculus  &  $\EXPTIME$ &  $\PSPACE$ \\ \hline
$\mu$-calculus, fixed formula &  $\EXPTIME$ &  $\PSPACE$ \\ \hline
$\mu$-calculus, fixed system &  $\EXPTIME$ &  $\PSPACE$ \\ \hline
$\CTL$, fixed formula & $\EXPTIME$  &  $\PSPACE$ (*)\\ \hline
$\CTL$, fixed system & $\EXPTIME$  &  $\PSPACE$ (*)\\ \hline
$\CTL$, fixed system, fixed leftward until depth & $\EXPTIME$  &  in $\P$ (*)\\ \hline
$\EF$           & $\PSPACE$  &  $\mathsf{P}^{\mathsf{NP}}$ (*)\\ \hline
$\EF$, fixed formula & $\PSPACE$ & $\mathsf{P}_{\log}^{\mathsf{NP}}$ hard\\ \hline
$\EF$, fixed system & $\PSPACE$ & in $\P$ 
\end{tabular}
\end{center}
\caption{Model checking over PDA and OCP{\label{F Model checking}; our new results are marked with (*).}}
\end{table}
{\em Markov decision processes} (MDPs) extend classical Markov chains by allowing so called
{\em nondeterministic vertices}. In these vertices,
no probability distribution on the outgoing
transitions is specified. The other vertices are called {\em probabilistic
vertices}; in these
vertices a probability distribution on the outgoing transitions is given.
The idea is that in an MDP a player Eve plays against nature (represented
by the probabilistic vertices).
In each nondeterministic vertex $v$, Eve chooses a probability distribution on the
outgoing transitions of $v$; this choice may depend on the past of the play (which is a path 
in the underlying graph ending in $v$) and is formally represented by a strategy for Eve.
An MDP together with a strategy for Eve 
defines an ordinary Markov chain, whose state space
is the unfolding of the graph underlying the MDP.  In Section~\ref{S Markov}, we consider infinite MDPs, which
are finitely represented by one-counter processes; this formalism was introduced in 
\cite{BraBroEteKucWoj09} under the name {\em one-counter Markov decision process} (OC-MDP).
For a given OC-MDP $\mathcal{A}$ and a set $R$ of control locations of the OCP underlying $\mathcal{A}$ 
(a so called {\em reachability constraint}) 
the following two sets $\ValOne(R)$ and $\OptValOne(R)$ were considered in 
\cite{BraBroEteKucWoj09}:
$\ValOne(R)$ is the set of all states $s$
of the MDP defined by $\mathcal{A}$ such that for every $\epsilon > 0$ there exists
a strategy $\sigma$ for Eve under which the
probability of finally reaching from $s$ a control location in $R$ and at
the same time having counter value $0$ is at least $1-\varepsilon$.
$\OptValOne(R)$ is the set of all states $s$
of the MDP defined by $\mathcal{A}$ for which there exists a specific strategy
for Eve under which this probability becomes $1$. It was shown in \cite{BraBroEteKucWoj09}
that for a given OC-MDP $\mathcal{A}$, a set of control locations $R$, and a state $s$
of the MDP defined by $\mathcal{A}$,
\begin{itemize}
\item the question whether $s \in \OptValOne(R)$  is  
$\PSPACE$-hard and in $\EXPTIME$, and
\item the question whether $s \in \ValOne(R)$  is  hard 
for every level of the Boolean hierarchy $\BH$.
\end{itemize}
The Boolean hierarchy is a hierarchy of complexity classes between $\NP$ and 
$\mathsf{P}^{\mathsf{NP}[\log]}$, see Section~\ref{S Tools} for a definition.
We use our lower bound techniques in order to improve the second hardness result for the levels of $\BH$ to 
$\PSPACE$-hardness. As a byproduct, we also reprove $\PSPACE$-hardness for $\OptValOne(R)$.
Currently, it is open, whether $\ValOne(R)$ is decidable; the corresponding problem for 
MDPs defined by pushdown processes is undecidable \cite{EtYa05}.

The paper is organized as follows. 
In Section~\ref{S Pre} we introduce general notation.
In Section~\ref{S OCP CTL} we define one-counter processes and the branching-time
logic $\CTL$. Periodicity of $\CTL$ on OCPs and a derived model checking algorithm 
is content of Section~\ref{S Upper}. 
In Section~\ref{S Expression} we give a fixed one-counter net (which is
basically  a one-counter process that cannot test if the counter is zero) for which
CTL model checking is $\PSPACE$-hard.
Section~\ref{S Tools} recalls tools from complexity theory that we need in
subsequent sections. We show that there already exists a fixed $\CTL$ formula
for which model checking over one-counter nets is $\PSPACE$-hard in Section~\ref{S Data}. 
Finally, we apply our lower bound technique and provide two applications.
We prove in Section~\ref{S Combined} that model checking
$\EF$ over one-counter nets is $\P^\NP$-hard, thus matching the $\P^\NP$ upper
bound from \cite{GoMaTo09}.
In Section~\ref{S Markov} we show that membership in $\ValOne(R)$
over one-counter Markov decision processes is $\PSPACE$-hard.

\section{Preliminaries}{\label{S Pre}}

\newcommand{\Q}{\mathbb{Q}}
\newcommand{\N}{\mathbb{N}}
\newcommand{\Z}{\mathbb{Z}}
\renewcommand{\O}{\mathcal{O}}
\newcommand{\bit}{\text{bit}}

We denote the naturals by $\N=\{0,1,2,\ldots\}$
and the rational numbers by $\Q$.
For each $i,j\in\N$ we define $[i,j]=\{k\in\N\mid
i\leq k\leq j\}$ and $[j]=[1,j]$. In particular $[0] = \emptyset$.
For each $n\in\N$ and each position $i\geq 1$, let $\bit_i(n)$ denote the $i^{\text{th}}$
least significant bit of the binary representation of $n$, i.e.,
$n=\sum_{i\geq 1} 2^{i-1}\cdot\bit_i(n)$.
For every finite and non-empty subset $M\subseteq\N\setminus\{0\}$, 
define $\LCM(M)$ to be the {\em least common multiple} of all numbers in $M$. 
Due to a result of Nair \cite{Nai82} it is known that 
$2^k\leq\LCM([k])\leq4^k$ for all $k\geq 9$.
As usual, for (a possibly infinite) alphabet $A$, $A^*$ denotes the set of all finite words
over $A$, $A^+$ denotes the set of all finite non-empty words over $A$, and 
$A^\omega$ denotes the set of all infinite words over $A$. Let $A^\infty = A^* \cup A^\omega$.
The length of a finite word 
$w$ is denoted by $|w|$.
For a word $w=a_1a_2\cdots a_n \in A^*$ (resp. $w=a_1a_2\cdots \in A^\omega$)
with $a_i\in A$ and $i \in [n]$ (resp. $i\geq 1$),
we denote by  
$w_i$ the $i^{\text{th}}$ letter $a_i$. 
A (possibly infinite) directed graph $G = (V, E)$ (with $E \subseteq V \times V$) is called
{\em deadlock-free} if for all $v \in V$ there exists $v' \in V$ with $(v,v') \in E$. If 
for all $v \in V$ there are only finitely many $v' \in V$ with $(v,v') \in E$, then $G$ is 
called {\em image-finite}.  The set of all {\em finite paths in $G$} is the set 
$\path_+(G) = \{ \pi \in V^+ \mid \forall i \in [|\pi|-1] : (\pi_i, \pi_{i+1}) \in E \}$. The 
set of all {\em infinite paths in $G$} is the set 
$\path_\omega(G) = \{ \pi \in V^\omega \mid \forall i \geq 1 : (\pi_i, \pi_{i+1}) \in E \}$.
A nondeterministic finite automaton (NFA) is a tuple
$A=(S,\Sigma,\delta,s_0,S_f)$, where $S$ is a finite set of {\em states},
$\Sigma$ is a {\em finite alphabet}, $\delta\subseteq S\times\Sigma\times S$ is the
{\em transition relation}, $s_0\in S$ is the {\em initial state}, and
$S_f\subseteq S$ is 
a set of {\em final states}.
We assume that the reader has some basic knowledge in complexity
theory, see e.g. \cite{AroBar09} for more details. 

\section{One-counter processes and computation tree logic}{\label{S OCP CTL}}
\newcommand{\Prop}{\mathcal{P}}

Fix some countable set $\Prop$ of {\em atomic propositions}.
A {\em transition system} is a triple $T=(S,\{S_p\mid p\in\Prop\}, \rightarrow)$, where
$(S, \to)$ is a directed graph and $S_p\subseteq S$ for all $p \in \Prop$ with $S_p = \emptyset$ for all but 
finitely many $p\in\Prop$. Elements of $S$ (resp. $\to$) are also
called {\em states} (resp. {\em transitions}).
We prefer to use the infix notation $s_1\rightarrow s_2$ 
instead of $(s_1,s_2)\in\,\rightarrow$. 
For $x \in \{+,\omega\}$ let $\path_x(T) = \path_x(S,\to)$.
For a subset $U\subseteq S$ of states, a (finite or infinite) path $\pi$ is called a {\em $U$-path}
if $\pi \in U^\infty$.

A {\em one-counter process} (OCP) is a tuple
$\O=(Q,\{Q_p\mid p\in\Prop\},\delta_0,\delta_{>0})$, where 
$Q$ is a finite set of {\em control locations}, 
$Q_p\subseteq Q$ for each $p\in\Prop$ but $Q_p = \emptyset$ for all but
finitely many $p\in\Prop$, 
$\delta_0\subseteq Q \times\{0,1\}\times Q$ is a finite set of
{\em zero transitions},  and 
$\delta_{>0}\subseteq Q\times \{-1,0,1\}\times Q$ is a finite set of
{\em positive transitions}. 
The {\em size} of an OCP is defined as
$|\O|=|Q|+\sum_{p\in\Prop}|Q_p| +|\delta_0|+|\delta_{>0}|$.
A {\em one-counter net} (OCN) is an OCP, where
$\delta_0\subseteq\delta_{>0}$.
A one-counter process
$\O=(Q,\{Q_p\mid p\in\Prop\},\delta_0,\delta_{>0})$ defines a transition system
$T(\O)=(Q\times\N,\{Q_p\times\N\mid p\in\Prop\},\rightarrow)$, where 
$(q,n)\rightarrow(q',n+k)$ if and only if either 
$n = 0$ and $(q,k,q')\in\delta_0$, or
$n > 0$ and $(q,k,q')\in\delta_{>0}$.

\newcommand{\X}{\mathsf{X}}
\newcommand{\F}{\mathsf{F}}
\renewcommand{\G}{\mathsf{G}}
\newcommand{\WU}{\mathsf{WU}}
\newcommand{\links}{[\![}
\newcommand{\rechts}{]\!]}
\newcommand{\sem}[1]{\ensuremath{\links #1 \rechts}}
\newcommand{\true}{\texttt{true}}
\newcommand{\false}{\texttt{false}}

More details on $\CTL$ and $\EF$ can be found 
for instance in \cite{BaiKat08}.
{\em Formulas} $\varphi$ of the logic $\CTL$ are given by the following
grammar, where $p\in\Prop$\!,
$$
\varphi\quad::=\quad p\ \mid\ \neg\varphi\ \mid\ \varphi\wedge\varphi\ \mid\ 
\exists\X\varphi\ \mid\ \exists\varphi\U\varphi\ \mid\ \exists\varphi\WU\varphi.
$$ 
Given a transition system $T=(S,\{S_p\mid p\in\Prop\},\rightarrow)$ 
and a $\CTL$ formula $\varphi$, we define 
the semantics $\links\varphi\rechts_T\subseteq S$ 
by induction on the structure of $\varphi$ as
follows:
\begin{eqnarray*}
\sem{p}_T&\  =\ & S_p\quad \text{for each } p\in\Prop\\
\sem{\neg\varphi}_T& = & S\setminus\sem{\varphi}_T\\
\sem{\varphi_1\wedge\varphi_2}_T&= &
\sem{\varphi_1}_T\cap\sem{\varphi_2}_T\\
\sem{\exists\X\varphi}_T& = & \{s\in S\mid \exists s'\in\sem{\varphi}_T: s\rightarrow
s'\}\ \\
\sem{\exists\varphi_1\U\varphi_2}_T &  = & 
\{s\in S\mid \exists  \pi \in \path_+(T) : \pi_1 = s, 
\pi_{|\pi|} \in\sem{\varphi_2}_T, \forall i\in[|\pi|-1] : \pi_i\in\sem{\varphi_1}_T\}\\
\sem{\exists\varphi_1\WU\varphi_2}_T &  = & 
\sem{\exists\varphi_1\U\varphi_2}_T\cup
\{s\in S\mid \exists \pi \in \path_\omega(T) : \pi_1 = s, 
\forall i\geq 1 : \pi_i\in\sem{\varphi_1}_T\}
\end{eqnarray*}
We write $(T,s)\models\varphi$ as an abbreviation for 
$s\in\sem{\varphi}_T$. When additionally $T$ is clear from the context, we just
write $s\models\varphi$. We introduce the usual abbreviations 
$\varphi_1\vee\varphi_2=\neg(\neg\varphi_1\wedge\neg\varphi_2)$,
$\true=p\vee\neg p$ for some $p\in\Prop$,
$\forall \X\varphi=\neg\exists\X\neg\varphi$,
$\exists\F\varphi=\exists\true\U\varphi$, and
$\exists\G\varphi=\exists\varphi\WU\false$.
Formulas of the $\CTL$-fragment $\EF$ are given by the following grammar, where
$p\in\Prop$,
$$
\varphi\quad::=\quad p\ \mid \neg\varphi\ \mid\ \varphi\wedge\varphi\ \mid\
\exists\X\varphi\ \mid \exists\F\varphi.
$$
Define the {\em size} $|\varphi|$ of $\CTL$ formulas $\varphi$ inductively as
follows: $|p|=1$,
$|\neg\varphi|=|\varphi|+1$, $|\varphi_1\wedge\varphi_2|=|\varphi_1|+|\varphi_2|+1$,
$|\exists\X\varphi|=|\varphi|+1$, and
$|\exists\varphi_1\U\varphi_2|=|\exists\varphi_1\W\U\varphi_2|=|\varphi_1|+|\varphi_2|+1$.

\section{$\CTL$ on OCPs: Periodic behaviour  and upper bounds}\label{S Upper}

\newcommand{\lud}{\mathrm{lud}}
The goal of this section is to prove a periodicity property of $\CTL$ over
one-counter processes. We will use this property in order to establish 
an upper bound for $\CTL$ on OCPs, see Theorem~\ref{CTL upper bound_0}. As a corollary, we show
that for a fixed one-counter process, $\CTL$ model checking restricted
to formulas of fixed leftward until depth (see the definition below)
can be done in polynomial time, see Corollary~\ref{CTL upper bound}.
For this, let us define the {\em leftward until depth $\lud$} of $\CTL$ formulas
inductively as follows: 
\begin{eqnarray*}
\lud(p)&\ =\ &0\ \text{ for each }p\in\Prop\\
\lud(\neg\varphi)&=&\lud(\varphi)\\
\lud(\varphi_1\wedge\varphi_2)&=&
\max\{\lud(\varphi_1),\lud(\varphi_2)\}\\
\lud(\exists\X\varphi)&=&\lud(\varphi)\\
\lud(\exists\varphi_1\U\varphi_2)&=&\max\{\lud(\varphi_1)+1,\lud(\varphi_2)\} \\
\lud(\exists\varphi_1\WU\varphi_2)&=&\max\{\lud(\varphi_1)+1,\lud(\varphi_2)\}\\
\end{eqnarray*}
A similar definition of the until depth can be found in \cite{TheWi96}, but
there  the until depth of $\exists\varphi_1\U\varphi_2$ is 
1 plus the maximum of the until depths of $\varphi_1$ and 
$\varphi_2$. Note that $\lud(\varphi)\leq 1$ for each $\EF$ formula $\varphi$.

Let us fix some one-counter process  $\O=(Q,\{Q_p\mid
p\in\Prop\},\delta_0,\delta_{>0})$ for the rest of this section. Let us introduce a
bit more notation.
Let $\odot\in\{+,-\}$, let $\delta\in\N$, and let 
$\pi=(q_1,n_1)\rightarrow(q_2,n_2)\ \cdots\rightarrow (q_k,n_k)$
(resp. $\pi=(q_1,n_1)\rightarrow(q_2,n_2)\rightarrow\cdots$)  be a finite
(resp. infinite) path in $T(\O)$ such that moreover $n_i,n_i\odot\delta>0$ for all $i$. 
Define $\pi\odot\delta$ to be the path that emerges from $\pi$ by replacing each $n_i$ by
$n_i\odot\delta$.
For each position $i$ and $j$ of $\pi$ with $i\leq j$, define $\pi[i,j]$ to be the
subpath of $\pi$ that begins in $(q_i,n_i)$ and that ends in $(q_j,n_j)$.

We aim to prove the following: For each $\CTL$ formula $\varphi$ we can compute
some threshold $t$ and some period $p$, where $t,p\in\exp(|\O|\cdot|\varphi|)$, such
that for all $n\in\N$ with $n>t$ only $n$'s residue class modulo $p$ determines whether
$(q,n)\in\sem{\varphi}_{T(\O)}$ or not, where $q\in Q$ is an arbitrary control
location.
The goal of this section is to give rather precise bounds on the size of
the threshold $t$ and the period $p$ embracing the notion of leftward until
depth from above.

Let us assume that $|Q|=k$. Define $K=\LCM([k])$ and
$K_\varphi=K^{\lud(\varphi)}$ for each $\CTL$ formula $\varphi$.

\begin{theorem} \label{thm-CTL-periodic}
Let $\varphi$ be a $\CTL$ formula. Then we can compute in polynomial time a threshold 
$$t(\varphi)\ \leq\ 2\cdot|\varphi|\cdot k^2\cdot K_\varphi$$ 
such that for all $n,n'>t(\varphi)$ that satisfy $n\equiv n' \text{ mod } K_\varphi$ we have
\begin{eqnarray}{\label{E period}}
(q,n)\in\sem{\varphi}_{T(\O)}\quad \text{\ if
and only if\ }\quad 
(q,n')\in\sem{\varphi}_{T(\O)}
\end{eqnarray} for each control location $q\in Q$. 
\end{theorem}

\begin{proof}
We prove the theorem by induction on the structure of $\varphi$.
That $t(\varphi)$ can be computed in polynomial time will be obvious.

\medskip
\noindent
Assume $\varphi\in\Prop$. Then we put $t(\varphi)=0$. Recall that
$K_\varphi=K^{\lud(\varphi)}=1$. 
Trivially, (\ref{E period}) holds.

\medskip
\noindent
Assume $\varphi=\neg\psi$. Then we put
$t(\varphi)=t(\psi)$. 
Equation (\ref{E period}) follows immediately by induction hypothesis.

\medskip
\noindent
Assume $\varphi=\psi_1\wedge\psi_2$.
Then we put
$t(\varphi)=\max\{t(\psi_1),t(\psi_2)\}$.
We have
\begin{eqnarray*}
t(\varphi)&\quad=\quad&\max\{t(\psi_1),t(\psi_2)\}\\
&\stackrel{\tiny\text{IH}}{\leq}&
\max\{2\cdot |\psi_i|\cdot k^2\cdot K_{\psi_i}\mid
i\in\{1,2\}\}\\
&\leq&
2\cdot|\varphi|\cdot k^2 \cdot K_\varphi
\end{eqnarray*}
and hence $t(\varphi)$ satisfies the requirement of the theorem.
Note that
$K_\varphi=\LCM\{K_{\psi_1},K_{\psi_2}\}$ by definition.
By choice of $t(\varphi)$, 
Equation (\ref{E period}) holds
immediately due to induction hypothesis.

\medskip
\noindent
Assume $\varphi=\exists\X \psi$. Then we put
$t(\varphi)=t(\psi)+K_\psi$.
Thus we get
\begin{eqnarray*}
t(\varphi)&\quad=\quad&t(\psi)+K_\psi\\
&\stackrel{\tiny\text{IH}}{\leq}&
2\cdot |\psi|\cdot k^2\cdot K_\psi+K_\psi\\
&\leq&
2\cdot (|\psi|+1)\cdot k^2\cdot K_\psi\\
&=&
2\cdot|\varphi|\cdot k^2\cdot K_\varphi
\end{eqnarray*} 
and hence $t(\varphi)$ satisfies the requirement of the theorem.
Since $t(\varphi)-t(\psi)=K_\psi\geq 1$, we have that (\ref{E period}) follows immediately by
induction hypothesis.

\medskip
\noindent
Assume $\varphi=\exists\psi_1\U\psi_2$. 
Let us first define the threshold.
Let $T=\max\{t(\psi_1),t(\psi_2)\}$.
We put $t(\varphi)=T+2\cdot k^2\cdot K_\varphi$.
Hence we have
\begin{eqnarray*}
t(\varphi)&\quad=\quad& T+2\cdot k^2\cdot K_\varphi\\
&\stackrel{\tiny\text{IH}}{\leq}&
\max\{2\cdot|\psi_i|\cdot k^2\cdot K_{\psi_i} \mid i\in\{1,2\}\}+
2\cdot k^2\cdot K_\varphi\\
&\leq&
2\cdot \left((|\varphi|-1)+1\right)\cdot k^2\cdot K_\varphi\\
&=&
2\cdot |\varphi|\cdot k^2\cdot K_\varphi
\end{eqnarray*}
and thus  $t(\varphi)$ satisfies the requirement of the theorem.
It remains to prove (\ref{E period}).

Recall that $K_\varphi=\LCM\{K\cdot K_{\psi_1},K_{\psi_2}\}$ by definition.
Let us fix an arbitrary control location $q\in Q$ and naturals $n,n'\in\N$
such that $t(\varphi)<n<n'$ and $n\equiv n' \text{ mod } K_\varphi$. We have
to prove that (\ref{E period}) holds, i.e., 
$(q,n)\in\sem{\varphi}_{T(\O)}$ if and only if 
$(q,n')\in\sem{\varphi}_{T(\O)}$.
For this, let $\delta=n'-n$, which is a multiple of $K_\varphi$.
The current situation is shown in Figure~\ref{F Situation}.

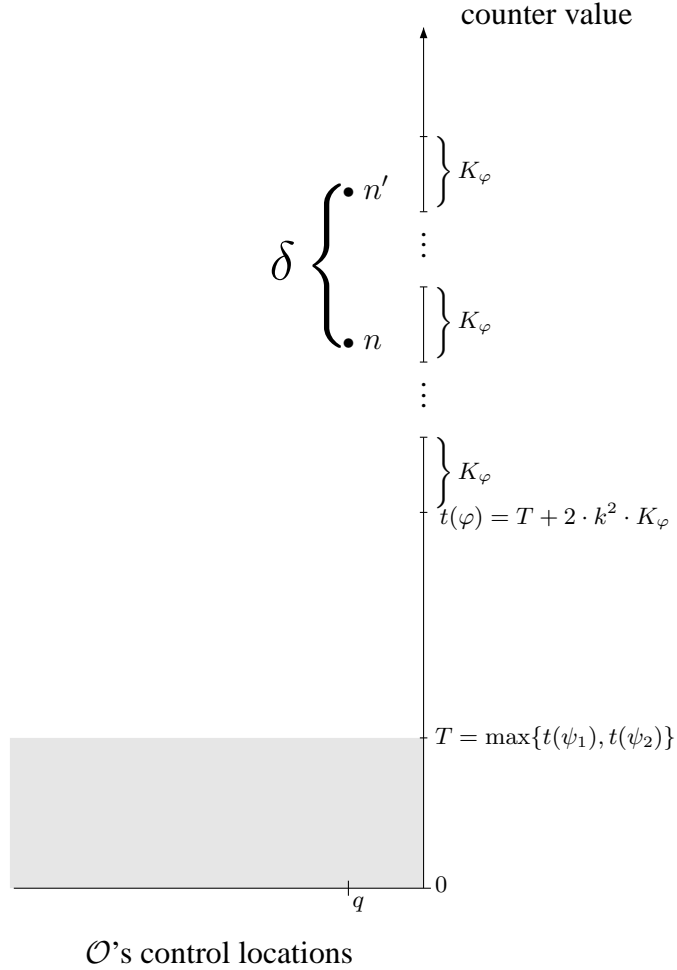
\begin{figure}[t]
\begin{center}
\setlength{\unitlength}{0.05cm}
\begin{picture}(200,250)(0,-20)

\gasset{Nframe=n,loopdiam=9,ELdist=.7}
\gasset{Nadjust=wh,Nadjustdist=1}
\gasset{curvedepth=0}
\node(li)(50,0){}		
\node(re)(163,0){}		
\node(u)(160,-1){}		
\node(o)(160,121){}		
\put(170,230){\large counter value}		
\put(163,-1){$0$}
\put(155,109){$\hspace{-0.7cm}\left.\phantom{\over{\bigcup_{\bigcup_{\bigcup}}}{\bigcup_\bigcup}}\right\} K_\varphi$}
\put(159,128){$\large\vdots$}
\put(70,-20){\large $\O$'s control locations}		
\node(u1)(160,139){}		
\node(o1)(160,161){}		
\put(155,149){$\hspace{-0.7cm}\left.\phantom{\over{\bigcup_{\bigcup_{\bigcup}}}{\bigcup_\bigcup}}\right\} K_\varphi$}
\put(159,168){$\large\vdots$}
\node(u2)(160,179){}		
\node(o2)(160,201){}		
\put(155,189){$\hspace{-0.7cm}\left.\phantom{\over{\bigcup_{\bigcup_{\bigcup}}}{\bigcup_\bigcup}}\right\} K_\varphi$}
\node(u3)(160,199){}		
\node(o3)(160,230){}		
\node(n)(140,145){$\small\bullet$}
\put(144,143.5){\large$n$}
\node(n')(140,185){$\small\bullet$}
\put(144,183.5){\large$n'$}
\put(119,162){\huge$\delta\left\{\huge\phantom{\over{\bigcup_{\bigcup}}{\bigcup_\bigcup}}\right.$}
\drawpolygon[Nframe=n,Nfill=y,fillgray=.9](50,0)(160,0)(160,40)(50,40)
\node(T)(195,40){$T=\max\{t(\psi_1),t(\psi_2)\}$}	
\drawline[AHnb=0](159,40)(161,40){}
\node(tp)(195,99){\small $t(\varphi)=T+2\cdot k^2\cdot K_\varphi$}	
\drawline[AHnb=0](159,100)(161,100){}
\drawline[AHnb=0](159,120)(161,120){}
\drawline[AHnb=0](159,140)(161,140){}
\drawline[AHnb=0](159,160)(161,160){}
\drawline[AHnb=0](159,180)(161,180){}
\drawline[AHnb=0](159,200)(161,200){}
\gasset{AHnb=0}   
\drawedge(li,re){}	 
\gasset{AHnb=1}
\drawedge[AHnb=0](u,o){}	 
\drawedge[AHnb=0](u1,o1){}	 
\drawedge[AHnb=0](u2,o2){}	 
\drawedge(u3,o3){}	 
\drawline[AHnb=0](140,2)(140,-2){}
\put(141,-5){$\large q$}
\end{picture}
\end{center}
\caption{The until case.\label{F Situation}}
\end{figure}

\medskip
\noindent
'Only-if': Let us assume that $(q,n)\in\sem{\varphi}_{T(\O)}$. Hence, there
exists a finite path
$$
\pi\ =\ (q_1,n_1)\rightarrow(q_2,n_2)\ \cdots\rightarrow (q_l,n_l), 
$$
where $l \geq 1$, $\pi[1,l-1]$ is a $\sem{\psi_1}_{T(\O)}$-path, $(q,n)=(q_1,n_1)$,
and $(q_{l},n_{l})\in\sem{\psi_2}_{T(\O)}$. 
Now we make a case distinction.

\medskip
\noindent 
{\em Case A:} $n_j>T$ for each $j\in[l]$. 
Since $K_{\psi_1}|\delta$ and $K_{\psi_2}|\delta$ we
obtain that
$\pi+\delta$ witnesses $(q,n')\in\sem{\varphi}_{T(\O)}$
by induction hypothesis. This is depicted in Figure~\ref{F Delta}.

\begin{figure}[t]
\begin{center}
\setlength{\unitlength}{0.05cm}
\begin{picture}(200,260)(0,-20)
\gasset{Nframe=n,loopdiam=9,ELdist=.7}
\gasset{Nadjust=wh,Nadjustdist=1}
\gasset{curvedepth=0}
\node(li)(50,0){}		
\node(re)(163,0){}		
\node(u)(160,-1){}		
\node(o)(160,121){}		
\put(170,230){\large counter value}		
\put(70,-20){\large $\O$'s control locations}		
\put(163,-1){$0$}
\put(155,109){$\hspace{-0.7cm}\left.\phantom{\over{\bigcup_{\bigcup_{\bigcup}}}{\bigcup_\bigcup}}\right\} K_\varphi$}
\put(159,128){$\large\vdots$}
\node(u1)(160,139){}		
\node(o1)(160,161){}		
\put(155,149){$\hspace{-0.7cm}\left.\phantom{\over{\bigcup_{\bigcup_{\bigcup}}}{\bigcup_\bigcup}}\right\} K_\varphi$}
\put(159,168){$\large\vdots$}
\node(u2)(160,179){}		
\node(o2)(160,201){}		
\put(155,189){$\hspace{-0.7cm}\left.\phantom{\over{\bigcup_{\bigcup_{\bigcup}}}{\bigcup_\bigcup}}\right\} K_\varphi$}
\node(u3)(160,199){}		
\node(o3)(160,230){}		
\node(n)(140,145){$\small\bullet$}
\put(144,143.5){\large$n$}
\node(n')(140,185){$\small\bullet$}
\put(144,183.5){\large$n'$}
\put(120,100){\large $\pi$}
\put(100,155){\large $\pi+\delta$}
\drawcurve(140,145)(120,90)(100,105)(80,62)(100,80)(60,102)(40,144)
\drawcurve[dash={0.4 1}0](140,185)(120,130)(100,145)(80,102)(100,120)(60,142)(40,184)
\drawpolygon[Nframe=n,Nfill=y,fillgray=.9](50,0)(160,0)(160,40)(50,40)
\drawline[AHnb=0](140,2)(140,-2){}
\put(141,-5){$\large q$}
\node(T)(195,40){$T=\max\{t(\psi_1),t(\psi_2)\}$}	
\drawline[AHnb=0](159,40)(161,40){}
\node(tp)(195,99){\small $t(\varphi)=T+2\cdot k^2\cdot K_\varphi$}	
\drawline[AHnb=0](159,100)(161,100){}
\drawline[AHnb=0](159,120)(161,120){}
\drawline[AHnb=0](159,140)(161,140){}
\drawline[AHnb=0](159,160)(161,160){}
\drawline[AHnb=0](159,180)(161,180){}
\drawline[AHnb=0](159,200)(161,200){}
\gasset{AHnb=0}
\drawedge(li,re){}	 
\gasset{AHnb=1}
\drawedge[AHnb=0](u,o){}	 
\drawedge[AHnb=0](u1,o1){}	 
\drawedge[AHnb=0](u2,o2){}	 
\drawedge(u3,o3){}	 
\end{picture}
\end{center}
\caption{The path $\pi+\delta$ witnesses $(q,n')\in\sem{\varphi}_{T(\O)}$.\label{F Delta}}
\end{figure}
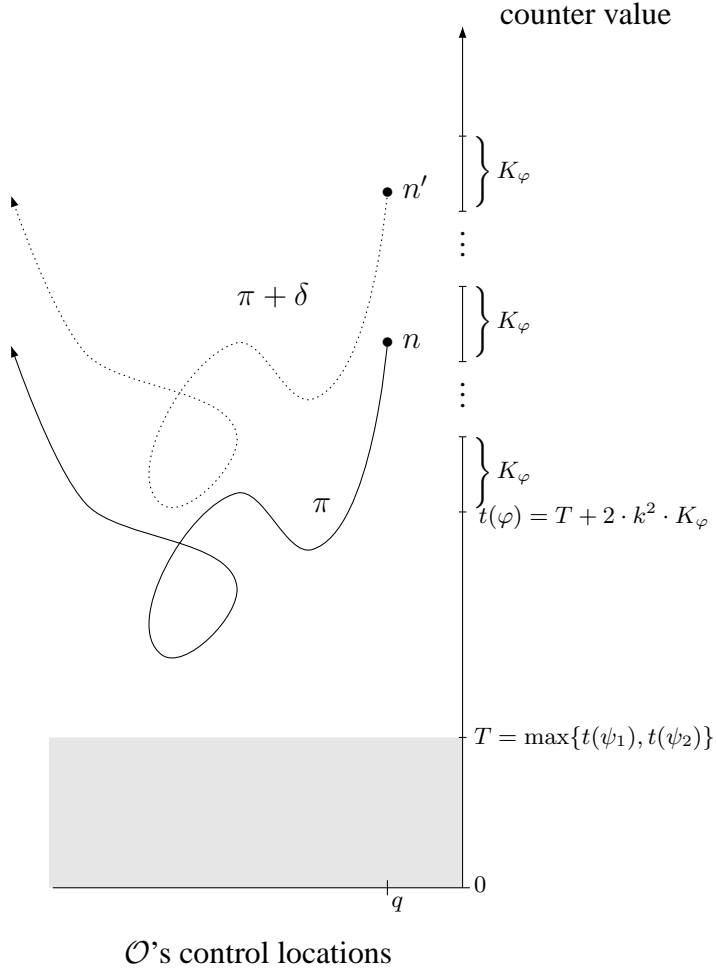

\medskip
\noindent 
{\em Case B:} $n_j\leq T$ for some $j\in[l]$. 
For each of $\pi$'s counter values $h\in\{n_i\mid i\in[l]\}$, define
$$
\mu(h)\quad=\quad\min\{i\in[l]\mid n_i=h\}
$$
to be the minimal position in $\pi$ whose corresponding state has counter value
$h$.
We are interested in $\pi$'s first states of counter value
$n, n-K_{\psi_1},n-2\cdot K_{\psi_1}$, and so on.
For this, define $m(i)=\mu(n-i\cdot K_{\psi_1})$ for every appropriate $i\in\N$.
By the pigeonhole principle, there are distinct $i_1,i_2\in[0,k]$ such that
$i_1 < i_2$ and 
$q_{m(i_1)}=q_{m(i_2)}$. 
Note that $i_1$ and $i_2$ are well-defined since 
$$n-i_1\cdot K_{\psi_1}\ >\  n-i_2\cdot K_{\psi_1} \ \geq\ n-k\cdot K_{\psi_1}\ \geq\ 
T+2\cdot k^2\cdot K_\varphi-k\cdot K_{\psi_1}\ >\  T.$$
Let $p = q_{m(i_1)}=q_{m(i_2)}$ and $d=i_1-i_2 \in [k]$.
Hence, $d$ divides $K$.
Moreover, let $\sigma$ denote $\pi$'s subpath from $(q_{m(i_1)},n_{m(i_1)})
= (p, n-i_1 \cdot K_{\psi_1})$ down to
$(q_{m(i_2)},n_{m(i_2)}) = (p, n-i_2 \cdot K_{\psi_1})
= (p, n-i_1 \cdot K_{\psi_1} - d \cdot K_{\psi_1} )$, i.e., formally
$\sigma=\pi[m(i_1),m(i_2)]$. 
The current situation is depicted in Figure~\ref{F Catch}.
The path $\sigma$ is indicated thick.

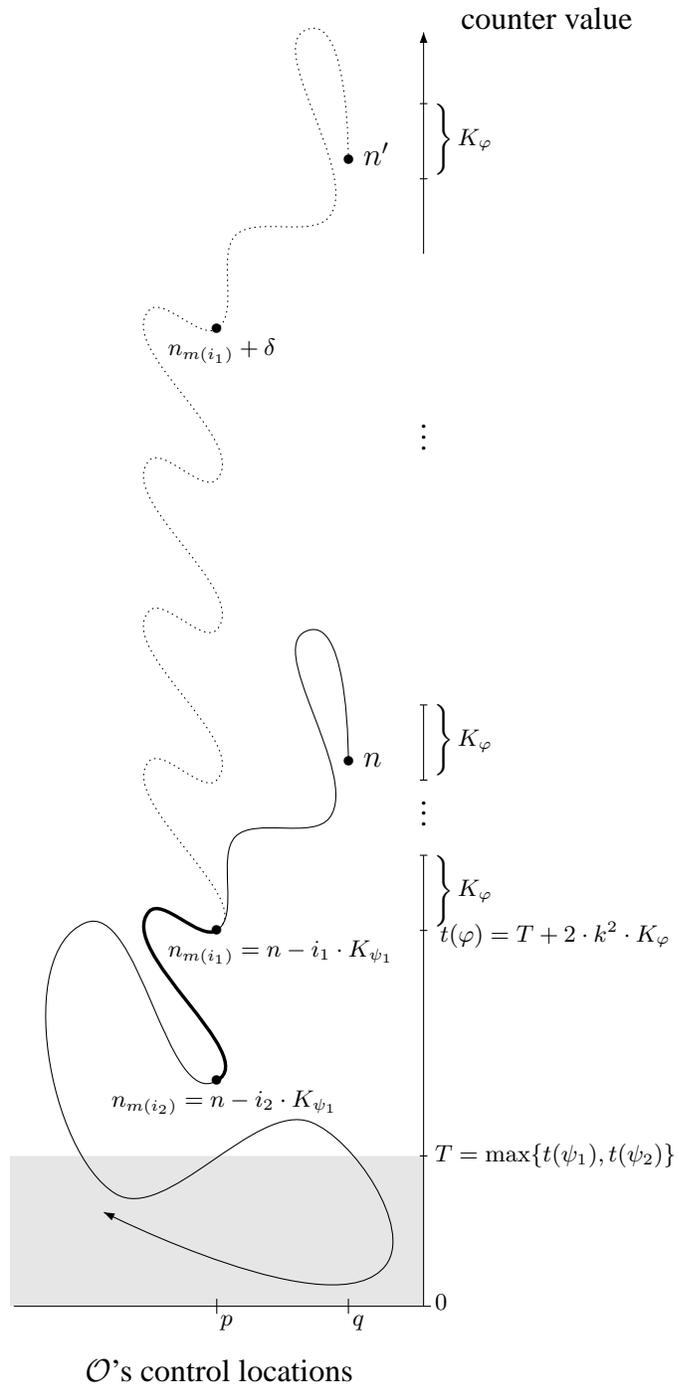
\begin{figure}
\begin{center}
\setlength{\unitlength}{0.05cm}
\begin{picture}(200,340)(0,-20)
\gasset{Nframe=n,loopdiam=9,ELdist=.7}
\gasset{Nadjust=wh,Nadjustdist=1}
\gasset{curvedepth=0}
\node(li)(50,0){}		
\node(re)(163,0){}		
\node(u)(160,-1){}		
\node(o)(160,121){}		
\put(170,340){\large counter value}		
\put(163,-1){$0$}
\put(155,109){$\hspace{-0.7cm}\left.\phantom{\over{\bigcup_{\bigcup_{\bigcup}}}{\bigcup_\bigcup}}\right\} K_\varphi$}
\put(159,128){$\large\vdots$}
\node(u1)(160,139){}		
\node(o1)(160,161){}		
\put(155,149){$\hspace{-0.7cm}\left.\phantom{\over{\bigcup_{\bigcup_{\bigcup}}}{\bigcup_\bigcup}}\right\} K_\varphi$}
\put(159,228){$\large\vdots$}
\node(u2)(160,209){}		
\node(o2)(160,241){}		
\put(155,309){$\hspace{-0.7cm}\left.\phantom{\over{\bigcup_{\bigcup_{\bigcup}}}{\bigcup_\bigcup}}\right\} K_\varphi$}
\node(u3)(160,279){}		
\node(o3)(160,340){}		
\node(n)(140,145){$\small\bullet$}
\put(144,143.5){\large$n$}
\node(n')(140,305){$\small\bullet$}
\put(144,303.5){\large$n'$}
\drawpolygon[Nframe=n,Nfill=y,fillgray=.9](50,0)(160,0)(160,40)(50,40)
\put(70,-20){\large $\O$'s control locations}		
\drawline[AHnb=0](140,2)(140,-2){}
\put(141,-5){$\large q$}
\drawline[AHnb=0](105,2)(105,-2){}
\put(106,-5){$\large p$}
\node(mi1)(105,100){$\small\bullet$}
\put(92,93){\small$n_{m(i_1)}= n-i_1\cdot K_{\psi_1}$}
\node(mi1')(105,260){$\small\bullet$}
\put(92,253){\small$n_{m(i_1)}+\delta$}
\node(mi2)(105,60){$\small\bullet$}
\put(77,53){\small$n_{m(i_2)}= n-i_2\cdot K_{\psi_1}$}
\drawcurve[AHnb=0,dash={0.4 1}0](140,305)(130,340)(135,290)(110,285)(105,260)
(87,265)(105,220)(87,225)(105,180)(87,185)(105,140)(87,145)(105,100)(104.9,100.1)
\drawcurve(140,145)(130,180)(135,130)(110,125)(105,100)(87,105)(105,60)(70,102)(78,30)(130,49)(150,11)(75,25)
\drawcurve[linewidth=0.9](105.2,101)(105,100)(87,105)(105,60)(104.7,60.42)
\node(T)(195,40){$T=\max\{t(\psi_1),t(\psi_2)\}$}	
\drawline[AHnb=0](159,40)(161,40){}
\node(tp)(195,99){\small $t(\varphi)=T+2\cdot k^2\cdot K_\varphi$}	
\drawline[AHnb=0](159,100)(161,100){}
\drawline[AHnb=0](159,120)(161,120){}
\drawline[AHnb=0](159,140)(161,140){}
\drawline[AHnb=0](159,160)(161,160){}
\drawline[AHnb=0](159,300)(161,300){}
\drawline[AHnb=0](159,320)(161,320){}
\gasset{AHnb=0}
\drawedge(li,re){}	 
\gasset{AHnb=1}
\drawedge[AHnb=0](u,o){}	 
\drawedge[AHnb=0](u1,o1){}	 
\drawedge(u3,o3){}	 
\end{picture}
\end{center}
\caption{The path from $(q,n)$ can be merged from $(q,n')$.\label{F Catch}}
\end{figure}

We have to prove $(q,n')\in\sem{\varphi}_{T(\O)}$.
For this, we show that there exists a $\sem{\psi_1}_{T(\O)}$-path
$\pi_\downarrow$ from $(q,n')$ down to $(q_{m(i_1)},n_{m(i_1)}) = (p, n-i_1 \cdot K_{\psi_1})$.
Thus, since $\pi_\downarrow$ meets $\pi$ in $(p,n-i_1\cdot K_{\psi_1})$, it follows 
$(q,n')\in\sem{\varphi}_{T(\O)}$. The path $\pi_\downarrow$ is indicated by a 
dashed  curve in Figure~\ref{F Catch}.
Our path $\pi_\downarrow$ consists of two concatenated paths.
First recall that the path $\sigma$ loses a counter height of precisely
$d\cdot K_{\psi_1}$.
The first part of $\pi_\downarrow$ is the path
$\pi[1,m(i_1)]$ shifted upwards by the offset $\delta$.
The second part of $\pi_\downarrow$ is the path from
$(q_{m(i_1)},n_{m(i_1)}+\delta) = (p, n-i_1 \cdot K_{\psi_1}+\delta)$ down to 
$(q_{m(i_1)},n_{m(i_1)}) = (p, n-i_1 \cdot K_{\psi_1})$ that we can obtain by
first shifting $\sigma$ up by the offset $\delta$ and then downward pumping it
precisely $\frac{\delta}{d\cdot K_{\psi_1}}$ many times.
This is possible since $\delta$ is a multiple of $K_\varphi$, which is in turn
a multiple of $K\cdot K_{\psi_1}$, hence $\frac{\delta}{d\cdot
K_{\psi_1}}\in\N$.

\medskip
\noindent
'If': Assume that $(q,n')\in\sem{\varphi}_{T(\O)}$.  To prove that
$(q,n)\in\sem{\varphi}_{T(\O)}$, we will use the following claim.

\medskip

\noindent
{\em Claim:}
Assume some $\sem{\psi_1}_{T(\O)}$-path
$(q_1,n_1) \to (q_2,n_2) \to \cdots \to (q_l,n_l)$ 
whose counter values are all
strictly above $T$ and where $n_1-n_l\geq k^2\cdot K\cdot K_{\psi_1}$. 
Then there exists a $\sem{\psi_1}_{T(\O)}$-path
from $(q_1,n_1)$ to $(q_l,n_l+K\cdot K_{\psi_1})$ strictly above $T+K\cdot
K_{\psi_1}$.
The statement of the claim is depicted in Figure~\ref{F Claim}.

\begin{figure}[t]
\begin{center}
\setlength{\unitlength}{0.05cm}
\begin{picture}(200,250)(0,-20)
\gasset{Nframe=n,loopdiam=9,ELdist=.7}
\gasset{Nadjust=wh,Nadjustdist=1}
\gasset{curvedepth=0}
\drawpolygon[Nframe=n,Nfill=y,fillgray=.9](40,0)(160,0)(160,40)(40,40)
\node(li)(40,0){}		
\node(re)(163,0){}		
\node(u)(160,-1){}		
\node(o)(160,121){}		
\put(170,230){\large counter value}		
\put(163,-1){$0$}
\node(u1)(160,139){}		
\node(o1)(160,161){}		
\node(u2)(160,179){}		
\node(o2)(160,201){}		
\node(u3)(160,199){}		
\node(o3)(160,230){}		
\node(T)(195,40){$T=\max\{t(\psi_1),t(\psi_2)\}$}	
\drawline[AHnb=0](159,40)(161,40){}
\node(tp)(195,99){\small $t(\varphi)=T+2\cdot k^2\cdot K_\varphi$}	
\drawline[AHnb=0](159,100)(161,100){}
\gasset{AHnb=0}   
\drawedge(li,re){}	 
\gasset{AHnb=1}
\drawedge(u,o3){}	 
\node(q1)(50,215){\tiny$\bullet$}
\put(50,220){$n_1$}
\node(ql)(110,88){\tiny$\bullet$}
\put(110,83){$n_l$}
\node(ql)(110,110){\tiny$\bullet$}
\put(110,115){$n_l+K\cdot K_{\psi_1}$}
\drawline[AHnb=0](50,2)(50,-2){}
\put(51,-5){$\large q_1$}
\drawline[AHnb=0](110,2)(110,-2){}
\put(111,-5){$\large q_l$}
\drawcurve[dash={0.4 2}0](50,215)(60,180)(90,150)(90,64)(110,110)
\drawline[AHnb=0,dash={0.9 5}0](40,215)(160,215){}
\drawline[AHnb=0,dash={0.9 5}0](40,88)(160,88){}
\drawline[AHnb=0,dash={0.9 5}0](40,62)(160,62){}
\drawline[AHnb=0,dash={0.9 5}0](40,110)(160,110){}
\drawcurve(50,215)(60,120)(70,150)(70,95)(50,110)(90,42)(110,88)
\put(119,162){}
\put(60,-20){\large $\O$'s control locations}		
\put(105,97){\large$\hspace{-0.7cm}\left.\phantom{\over{\bigcup_{\bigcup_{\bigcup}}}{\bigcup_\bigcup}}\right\}K\cdot K_{\psi_1}$}
\put(105,49){\large$\hspace{-0.7cm}\left.\phantom{\over{\bigcup_{\bigcup_{\bigcup}}}{\bigcup_\bigcup}}\right\}K\cdot K_{\psi_1}$}
\put(-5,150){$\geq k^2 \cdot K \cdot K_{\psi_1}$ \large$\begin{cases}
  \phantom{ \bigcup} & \\ 
  \phantom{ \bigcup} & \\ 
  \phantom{ \bigcup} & \\ 
  \phantom{ \bigcup} & \\ 
  \phantom{ \bigcup} & \\ 
  \phantom{ \bigcup} & \\ 
  \phantom{ \bigcup} & \\ 
  \phantom{ \bigcup} & \\ 
  \phantom{ \bigcup} & \\ 
  \phantom{ \bigcup} & \\ 
  \phantom{ \bigcup} &  
  \end{cases}$}
\end{picture}
\end{center}
\caption{Shortening paths above $T$ of height difference at least 
$k^2 \cdot K \cdot K_{\psi_1}$ by height $K\cdot K_{\psi_1}$.\label{F Claim}}
\end{figure}
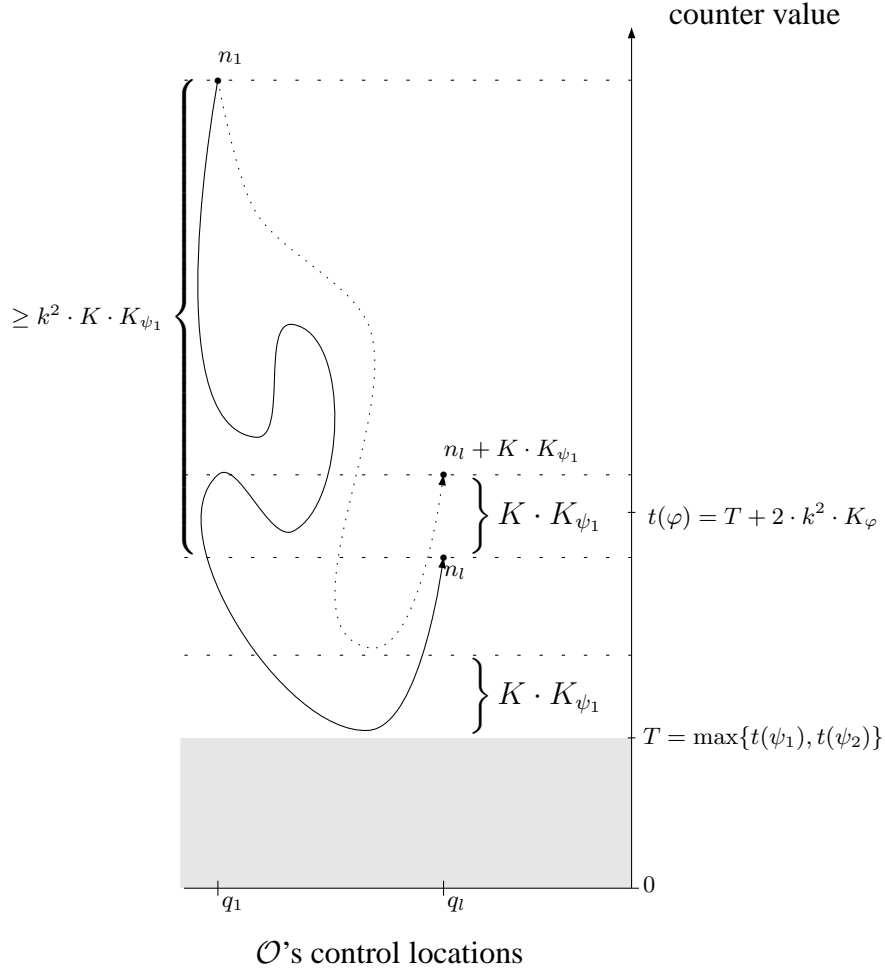

\vspace{-0.1cm}

\bigskip

\noindent
Thus, the claim tells us that paths that lose height
at least $k^2\cdot K\cdot K_{\psi_1}$ and
whose states all have counter values strictly
above $T$ can be lifted by a height
precisely $K\cdot K_{\psi_1}$.

Let us postpone the proof of the claim and first finish the proof of the 
if-direction. Since by assumption $(q,n')\in\sem{\varphi}_{T(\O)}$, there exists
a finite path
$$
\pi\ =\ (q_1,n_1)\rightarrow(q_2,n_2)\ \cdots\rightarrow (q_l,n_l), 
$$
where $\pi[1,l-1]$ is a $\sem{\psi_1}_{T(\O)}$-path, $(q,n')=(q_1,n_1)$, 
and where $(q_{l},n_{l})\in\sem{\psi_2}_{T(\O)}$. 
To prove $(q,n)\in\sem{\varphi}_{T(\O)}$, we make a case distinction.

\medskip
\noindent
{\em Case A:} $n_j>T$ for each $j\in[l]$. 
Assume that the path $\pi[1,l-1]$ contains two states whose counter difference
is at least $k^2\cdot K\cdot K_{\psi_1}+K_\varphi$ which is (strictly) greater
than $k^2\cdot K\cdot K_{\psi_1}$. Since $K_\varphi$ is a multiple of $K\cdot
K_{\psi_1}$ by definition, we can shorten $\pi[1,l-1]$ by a height precisely
$K_\varphi$ by applying the above claim $\frac{K_\varphi}{K\cdot
K_{\psi_1}}\in\N$
many times. We repeat this shortening process of $\pi[1,l-1]$ by height
$K_\varphi$ as long as this is no longer possible, i.e., until
there are no two states whose counter difference is at least
$k^2\cdot K \cdot K_{\psi_1}+K_\varphi$.
Let $\sigma$ denote the $\sem{\psi_1}_{T(\O)}$-path starting in $(q,n')$ that we obtain
from $\pi[1,l-1]$ until the before mentioned shortening is no longer possible.
Thus, $\sigma$ ends in some state
with a counter value that is congruent $n_{l-1}$ modulo $K_\varphi$ (since we
shortened $\pi[1,l-1]$ by a multiple of $K_\varphi$). 
Since $K_\varphi$ is in turn a multiple of $K_{\psi_2}$, we can
build a path $\sigma'$ which extends the path $\sigma$ by a single
transition to some state that satisfies $\psi_2$ by induction hypothesis. 
Moreover, by our shortening process,
the counter difference between any two states in $\sigma'$ is at most
$$k^2\cdot K\cdot K_{\psi_1}+K_\varphi\quad\leq\quad 2\cdot k^2\cdot
K_\varphi.$$ 

Since $n > T+2\cdot k^2\cdot K_\varphi$, it follows that the path
$\sigma'-\delta$ (which starts in $(q,n)$)
is strictly above $T$. Moreover, since $\delta$ is a multiple of 
$K_{\psi_1}$ and $K_{\psi_2}$, this path
witnesses $(q,n)\in\sem{\varphi}_{T(\O)}$ by induction hypothesis.

\medskip
\noindent
{\em Case B:} $n_j=T$ for some $j\in[l]$. Let $j_0\in[l]$ be minimal such that
$n_{j_0}=T$.
Note that $\pi[1,j_0-1]$ is a $\sem{\psi_1}_{T(\O)}$-path whose counter
values are all strictly above $T$.
Moreover the maximal counter difference between two states of $\pi[1,j_0-1]$ is
at least
$$
2\cdot k^2\cdot K_\varphi-1+\delta\quad\geq\quad k^2\cdot K\cdot
K_{\psi_1}+\delta.
$$
Hence, in analogy to case A, we can shorten $\pi[1,j_0-1]$ {\em precisely} by
height $\delta$. Let $\sigma$ denote the resulting path. Then 
$\sigma-\delta$ is a $\sem{\psi_1}_{T(\O)}$-path that ends in 
$(q_{j_0-1}, n_{j_0-1})$ and starts in $(q,n)$. We can append $\pi[j_0-1,l]$
to this path. The resulting path 
witnesses $(q,n)\in\sem{\varphi}_{T(\O)}$.

\medskip
\noindent
It remains to prove the above claim.

\medskip
\noindent
{\em Proof of the claim.}
For each counter value $h\in\{n_i\mid i\in[l]\}$ that appears in $\pi$, let
$$\mu(h)=\min\{i\in[l]\mid n_i=h\}$$ denote the minimal position
in $\pi$ whose corresponding state has counter value $h$.
Define $\Delta=k\cdot K_{\psi_1}$. We will be interested in
$k\cdot K$ many consecutive intervals (of counter values) each of size $\Delta$
-- we will call these intervals blocks.
Define the bottom $b=n_1-(k\cdot K)\cdot\Delta$. A {\em block} is an interval
$B_i=[b+(i-1)\cdot\Delta,b+i\cdot\Delta]$ for some $i\in[k\cdot K]$. 
Since each block has size $\Delta=k\cdot K_{\psi_1}$, we can think of each block
$B_i$ to consist of $k$ consecutive {\em subblocks} of size $K_{\psi_1}$ each.
Note that each subblock has two extremal elements, namely its {\em upper} and
{\em lower boundary}. Thus all $k$ subblocks have $k+1$ boundaries in total.
Hence, by the pigeonhole principle, 
for each block $B_i$, there exists some distance $d_i\in[k]$ and 
two distinct boundaries $\beta(i,1)$ and $\beta(i,2)$ of distance $d_i\cdot
K_{\psi_1}$ such that the control location of $\pi$'s earliest state
of counter value $\beta(i,1)$ agrees with the control location of $\pi$'s
earliest state of counter value $\beta(i,2)$, i.e., formally
$$
q_{\mu(\beta(i,1))}\ =\ q_{\mu(\beta(i,2))}.
$$
The situation is depicted in Figure~\ref{F Blocks}. Observe that shortening the path $\pi$ by gluing together $\pi$'s states
at position $\mu(\beta(i,1))$ and $\mu(\beta(i,2))$ still results in a
$\sem{\psi_1}_{T(\O)}$-path by induction hypothesis, since we shortened the
height of $\pi$ by a multiple of $K_{\psi_1}$.
Our overall goal is to shorten $\pi$ by gluing together states only of certain 
blocks such that we obtain a path whose height is in total precisely
$K\cdot K_{\psi_1}$ smaller than $\pi$'s.

Recall that there are $k\cdot K$ many blocks.
By the pigeonhole principle there is some $d\in[k]$ such that $d_i=d$ for at
least $K$ many blocks $B_i$. By gluing together $\frac{K}{d} \in \mathbb{N}$ pairs of states of distance
$d\cdot K_{\psi_1}$ each, we shorten $\pi$ by a height of $\frac{K}{d}\cdot
d\cdot K_{\psi_1}=K\cdot K_{\psi_1}$. This proves the claim.

\begin{figure}[t]
\begin{center}
\setlength{\unitlength}{0.05cm}
\begin{picture}(200,280)(0,-20)
\gasset{Nframe=n,loopdiam=9,ELdist=.7}
\gasset{Nadjust=wh,Nadjustdist=1}
\gasset{curvedepth=0}
\put(170,230){\large counter value}		
\put(163,-1){$0$}
\node(li)(50,0){}		
\node(re)(163,0){}		
\node(u)(160,-1){}		
\node(o3)(160,230){}		
\gasset{AHnb=0}  
\drawedge(li,re){}	 
\gasset{AHnb=1}
\drawedge(u,o3){}	 
\put(49,29){
\drawcurve[AHnb=0](50,180)(80,200)(70,160)(50,175)(70,130)(65,120)(40,110)(45,100)(55,90)(60,80)(65,70)(44,50)(54,69)(70,49)(70,35)
\put(50,180){\tiny$\bullet$}
\put(47,185){$n_1$}
\drawline[AHnb=0](40,180)(80,180){}
\drawline[AHnb=0,dash={0.4 2}0](40,170)(80,170){}
\drawline[AHnb=0,dash={0.4 2}0](40,160)(80,160){}
\drawline[AHnb=0,dash={0.4 2}0](40,150)(80,150){}
\drawline[AHnb=0,dash={0.4 2}0](40,140)(80,140){}
\drawline[AHnb=0,dash={0.4 2}0](40,130)(80,130){}
\drawline[AHnb=0](40,120)(80,120){}
\put(-6,147){$\Delta=k\cdot K_{\psi_1}$\huge$\begin{cases} & \\ & \\ & \\ \end{cases}$}
\put(0,135){block $B_1$}
\put(53,143){$\hspace{-1cm}\Large \left.\phantom{\begin{cases} \bigcup& \\ &
   \bigcup\\    \end{cases}}\right\} d_1\cdot
    K_{\psi_1}$}
\put(64,93){$\hspace{-1cm} \left.\phantom{\begin{cases} & \\ & \\ & \\ &\\ &
   \\ & 
   \bigcup\\    \end{cases}}\right\} d_2\cdot
    K_{\psi_1}$}
\put(88,173){$\hspace{-0.7cm}\large \left.\phantom{\bigcup}\right\} K_{\psi_1}$}
\drawline[AHnb=0](40,120)(80,120){}
\drawline[AHnb=0,dash={0.4 2}0](40,110)(80,110){}
\drawline[AHnb=0,dash={0.4 2}0](40,100)(80,100){}
\drawline[AHnb=0,dash={0.4 2}0](40,90)(80,90){}
\drawline[AHnb=0,dash={0.4 2}0](40,80)(80,80){}
\drawline[AHnb=0,dash={0.4 2}0](40,70)(80,70){}
\drawline[AHnb=0](40,60)(80,60){}
\put(3,86){block $B_2$\huge$\begin{cases} & \\ & \\ & \\ \end{cases}$}
\put(30,-45){\large $\O$'s control locations}		
\put(69,26){$\vdots$}
    }
\end{picture}
\caption{Repeating control locations in blocks\label{F Blocks}}
\end{center}
\end{figure}
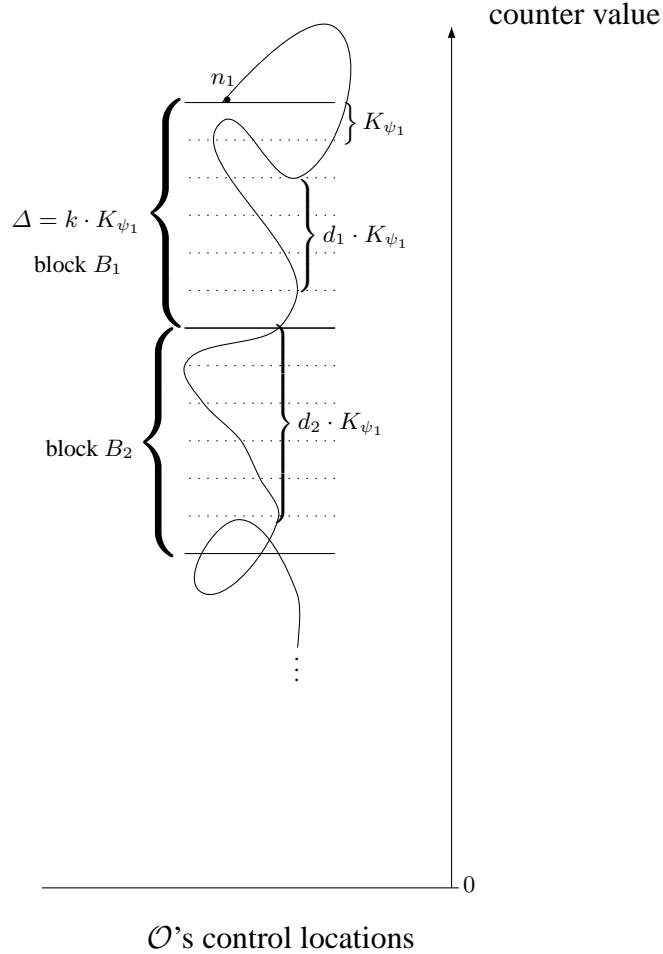

\medskip
\noindent
Assume $\varphi=\exists\psi_1\WU\psi_2$. This can can easily seen to be proven
analogously to the case when $\varphi=\exists \psi_1\U\psi_2$.
\qed
\end{proof}

\begin{theorem} \label{CTL upper bound_0}
The following problem can be solved in time 
$O(\log(n) +|Q|^3 \cdot |\varphi|^2 \cdot 4^{|Q| \cdot \lud(\varphi)} \cdot |\delta_0 \cup \delta_{>0}|)$:

\noindent
INPUT: A one-counter process $\O=(Q,\{Q_p\mid p\in\Prop\},\delta_0,\delta_{>0})$,
a $\CTL$ formula $\varphi$, a control location $q\in Q$ and some natural $n\in\N$ given in binary.

\noindent
QUESTION: $(q,n)\in\sem{\varphi}_{T(\O)}$?
\end{theorem}

\begin{proof}
Let $k = |Q|$. 
We first compute the threshold $t(\varphi) \leq 2 \cdot |\varphi| \cdot k^2
\cdot K_\varphi$ from  Theorem~\ref{thm-CTL-periodic}. Then we have $(q,n)\in\sem{\varphi}_{T(\O)}$
if and only if $(q,m)\in\sem{\varphi}_{T(\O)}$, where 
either $n=m \leq t(\varphi)$ or $n > t(\varphi)$ and
$m$ is the unique number in the interval $[t(\varphi)+1,
t(\varphi)+K_\varphi]$, which is congruent $n$ modulo $K_\varphi$. We can find
this number in time $O(\log(n))$. Now we check $(q,m)\in\sem{\varphi}_{T(\O)}$
using the standard algorithm for model checking $\CTL$ on finite transition 
systems. The only difference is that if we reach a counter value of
$t(\varphi)+K_\varphi+1$, then we replace this value by
$t(\varphi)+1$. More precisely, we compute inductively for every subformula
$\psi$ of $\varphi$ the set 
$$
S(\psi) = \sem{\psi}_{T(\O)} \cap (Q \times [t(\varphi)+K_\varphi]).
$$
Let us sketch the case of an until formula
$\psi =\exists \psi_1\U\psi_2$. By induction, we have already
computed the sets $S(\psi_1)$ and $S(\psi_2)$.
The set $S(\psi)$ is computed by a fixpoint iteration. Initially, 
we put all elements from $S(\psi_2)$ into $S(\psi)$.
Then, we perform the following fixpoint iteration process
as long as possible. Assume that $(p,k) \in S(\psi_1)$ 
is a state, which does not belong to the current $S(\psi)$. Assume 
that $(p,k)$ has a $T(\O)$-successor (where a counter value of 
$t(\varphi)+K_\varphi+1$ is reduced to $t(\varphi)+1$) in $S(\psi)$. Then we add
$(p,k)$ to $S(\psi)$. The correctness of this fixpoint iteration
process follows from Theorem~\ref{thm-CTL-periodic}. 
The size of each set $S(\psi)$ is bounded by 
$O(|Q| \cdot |\varphi| \cdot k^2 \cdot K_\varphi) \subseteq 
O(|Q|^3 \cdot |\varphi| \cdot 4^{|Q| \cdot \lud(\varphi)})$.
Computing $S(\psi)$ can be done in time 
$O(|Q|^3 \cdot |\varphi| \cdot 4^{|Q| \cdot \lud(\varphi)} \cdot |\delta_0
\cup \delta_{>0}|)$. Hence, the total time bound is
$O(\log(n) + |Q|^3 \cdot |\varphi|^2 \cdot 4^{|Q| \cdot \lud(\varphi)} \cdot
|\delta_0 \cup \delta_{>0}|)$.
\qed
\end{proof}

\begin{corollary} \label{CTL upper bound}
For every fixed one-counter process $\O=(Q,\{Q_p\mid
p\in\Prop\},\delta_0,\delta_{>0})$
and every fixed $k$ the following problem is in $\mathsf{P}$:

\noindent
INPUT: A $\CTL$ formula $\varphi$ with $\lud(\varphi)\leq k$, a control location
$q\in Q$ and some natural $n\in\N$ given in binary.

\noindent
QUESTION: $(q,n)\in\sem{\varphi}_{T(\O)}$?
\end{corollary}
Corollary~\ref{CTL upper bound} generalizes a result from \cite{GoMaTo09},
stating that the expression complexity of $\EF$ over one-counter processes
is in $\mathsf{P}$.

\section{Expression complexity for $\CTL$ is hard for $\PSPACE$}
{\label{S Expression}}

The goal of this section is to prove that model checking $\CTL$ is $\PSPACE$-hard
already over a fixed one-counter net.
We show this via a reduction from the well-known $\PSPACE$-complete problem
QBF.
Our lower bound proof is separated into three steps. In step one, we 
define a family of $\CTL$ formulas $(\varphi_i)_{i\geq 1}$ such that over the
fixed the one-counter net $\O$ that is depicted in Figure~\ref{Fig-fixed-CCP} we can express (non-)divisibility by $2^i$.
In step two, we define a family of $\CTL$ formulas $(\psi_i)_{i\geq 1}$ such that
over $\O$ we can express if the $i^{\text{th}}$
bit in the binary representation of a natural is set to $1$.
In our final step, we give the reduction from QBF. 

For step one, we need the following simple fact which characterizes divisibility by
powers of two. Recall that $[n] = \{1,\ldots,n\}$, in particular $[0] = \emptyset$.

\begin{fact}{\label{F div}}
Let $n \geq 0$ and $i\geq 1$. Then the following two statements are equivalent:
\begin{itemize}
\item $2^i$ divides $n$.
\item $2^{i-1}$ divides $n$ and $|\{n'\in [n] \mid  2^{i-1} \text{ divides } n'\}|\  \text{ is even.}$
\end{itemize}
\end{fact}
\begin{figure}[t]
\begin{center}
\setlength{\unitlength}{0.05cm}
\begin{picture}(300,133)(-90,-10)
  \put(-65,40){\huge$\delta_{>0}:$}
 \gasset{Nframe=n,loopdiam=9,ELdist=.7}
 \gasset{Nadjust=wh,Nadjustdist=1}
 \gasset{curvedepth=0}
 \node(tb)(45,80){$\overline{t}$}
 \node(t)(45,40){$t$}
 \node(q0)(100,20){$q_0$}
 \node(q2)(100,100){$q_2$}
 \node(q1)(135,60){$q_1$}
 \drawloop[loopangle=0](q1){$-1$}	
 \node(q3)(65,60){$q_3$}
 \drawloop[ELpos=66,loopangle=240](q3){$-1$}	
 \drawedge[ELside=l,ELpos=70](q0,q1){$-1$}
 \drawedge[ELside=r](q1,q2){$-1$}
 \drawedge[ELside=r](q2,q3){$-1$}
 \drawedge[ELside=l](q3,q0){$-1$}
 \node(f)(10,50){$f$}
 \node(g)(10,90){$g$}
 \drawedge[AHnb=1,ATnb=1,ELside=l,curvedepth=0](q0,t){$0$}
 \drawedge[AHnb=1,ATnb=1,ELside=l,ELpos=30,curvedepth=0](q1,tb){$0$}
 \drawedge[AHnb=0,ATnb=1,ELside=r,curvedepth=0](q2,tb){$0$}
 \drawedge[AHnb=1,ATnb=1,ELside=l,curvedepth=0](q3,tb){$0$}
 \drawedge[AHnb=1,ELside=l,curvedepth=17](q2,t){$0$}
 \drawedge[ELside=l,curvedepth=0](t,f){$0$}
 \drawedge[ELside=l,curvedepth=0](tb,f){$-1$}
 \drawedge[ELside=l,curvedepth=0,AHnb=1,ATnb=1](g,f){$-1$}
 \node(p0)(25,120){$p_0$}
 \node(p1)(65,120){$p_1$}
 \drawedge[ELside=l,ELpos=50,curvedepth=5](tb,p1){$+1$}
 \drawedge[ELside=l,ELpos=50,curvedepth=1](p1,tb){$0$}
 \drawedge[ATnb=1,AHnb=1,ELside=r](p0,tb){$0$}
 \drawloop[loopdiam=9,loopangle=0](p1){$+1$}	
\end{picture}
\end{center}
\begin{center}
\setlength{\unitlength}{0.05cm}
\begin{picture}(300,40)(-30,-15)
 \put(5,15){\huge$\delta_{0}:$}
 \gasset{Nframe=n,loopdiam=9,ELdist=.7}
 \gasset{Nadjust=wh,Nadjustdist=1}
 \gasset{curvedepth=0}
 \node(tb)(105,0){$\overline{t}$}
 \node(t)(165,10){$t$}
 \node(q0)(190,0){$q_0$}
 \node(f)(140,20){$f$}
 \drawedge[ELside=r](t,q0){$0$}
 \drawedge(t,f){$0$}
 \node(p0)(85,40){$p_0$}
 \node(p1)(125,40){$p_1$}
   \drawedge[ELside=l,ELpos=50,curvedepth=0](tb,p1){$+1$}
   \drawedge[ATnb=1,AHnb=1,ELside=r](p0,tb){$0$}
\end{picture}
\end{center}
\caption{\label{Fig-fixed-CCP} The one-counter net $\O$ for which $\CTL$ model checking
is $\PSPACE$-hard}
\end{figure}
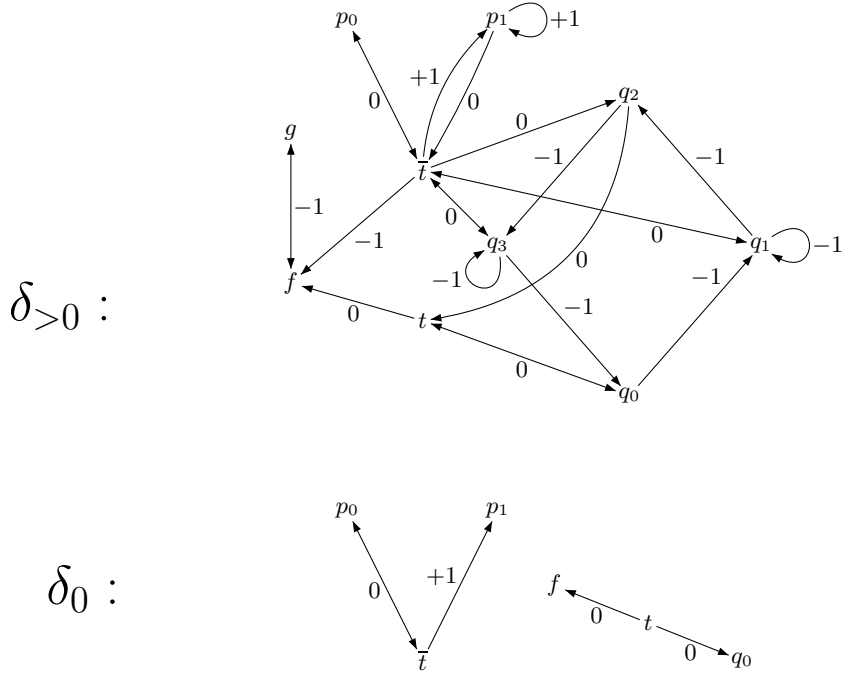
The set of atomic propositions of $\O$ in Figure~\ref{Fig-fixed-CCP}  
coincides with its control locations.
Recall that $\O$'s zero transitions are denoted by $\delta_0$ and $\O$'s positive
transitions are denoted by $\delta_{>0}$. Since $\delta_0\subseteq\delta_{>0}$, we have that
$\O$ is indeed a one-counter net.

Note that both $t$ and $\overline{t}$ are control locations of $\O$.
Now we define a family of $\CTL$ formulas
$(\varphi_i)_{i\geq 1}$ such
that for each $n\in\N$ we have that first $(t,n)\models\varphi_i$ if
and only if $2^i$ divides $n$ 
and second $(\overline{t},n)\models\varphi_i$ if and only if 
$2^i$ does {\em not} divide $n$.
On first sight, it might seem superfluous to let the control location $t$
represent divisibility by powers of two and the control location $\overline{t}$ to represent
non-divisibility by powers of two since $\CTL$ allows negation. 
However the fact that we have {\em only one} family of formulas
$(\varphi_i)_{i\geq 1}$ to
express both divisibility and non-divisibility is a crucial technical subtlety
that is necessary in order to avoid an exponential blowup in formula size.
By making use of Fact~\ref{F div}, we construct 
the formulas $\varphi_i$ inductively.
\newcommand{\test}{\text{test}}
First, let us define the auxiliary formulas $\test=t\vee\overline{t}$ and
$\varphi_\diamond=q_0\vee q_1\vee q_2\vee q_3$.
Think of $\varphi_\diamond$ to hold in those control locations that altogether
are situated in the ``diamond'' in Figure~\ref{Fig-fixed-CCP}.
We define
$$\varphi_1\quad =\quad \test\ \wedge\
\exists\X\left( f\wedge\EF(f\wedge\neg\exists\X g)\right).
$$
Now assume $i>1$. Then we define
\begin{eqnarray*}
\varphi_i & = & \test\ \wedge\ \exists\X\mu_i , \text{ where } \\
\mu_i & = &  \exists (\varphi_\diamond\wedge
           \exists\X\varphi_{i-1}) \U (q_0\wedge\neg\exists\X q_1) .
\end{eqnarray*}
Observe that $\varphi_i$ can only be true either in control location $t$ or
$\overline{t}$. 
Note that the formula right to the until symbol expresses that we are in $q_0$
and that the current counter value is zero.
Also note that the formula left to the until symbol
requires that $\varphi_\diamond$ holds, i.e., we are always in one of the four
``diamond control locations''. In other words, we decrement the counter
by moving along the diamond control locations (by possibly looping) and always check if
$\exists\X\varphi_{i-1}$ holds, just until we are in $q_0$ and the counter value
is zero. Since $\varphi_{i-1}$ is only used once in $\varphi_i$, we get:

\begin{fact}{\label{F size varphi}} $|\varphi_i|\in O(i)$.
\end{fact}
The following lemma shows the correctness of the construction.

\begin{lemma}{\label{L Correctness}}
Let $n\geq 0$ and $i\geq 1$. Then 
\begin{enumerate}[(1)]
\item $(t,n)\models\varphi_i$ if and only if $2^i$ divides $n$.
\item $(\overline{t},n)\models\varphi_i$ if and only if $2^i$ does
not divide $n$.
\end{enumerate}
\end{lemma}
\begin{proof}
We prove statements (1) and (2) simultaneously by induction on $i$.
For the induction base, assume $i=1$. We only show Point (2), i.e.
$(\overline{t},n)\in\sem{\varphi_1}_{T(\O)}$ if and only if $n$ is odd.
We have the following equivalences:
\begin{eqnarray*}
(\overline{t},n)\models\varphi_1
&\quad\Longleftrightarrow\quad&
n\geq 1\text{ and } (f,n-1) \models \EF(f\wedge\neg\exists\X g) \\
&\Longleftrightarrow&
n\geq 1\ \text{and}\ (f,n-1)\rightarrow^*(f,0)\\
&\Longleftrightarrow&n\geq 1\ \text{and}\  n-1 \text{ is even}\\
&\Longleftrightarrow&n \text{ is odd}
\end{eqnarray*}
Point (1) can be shown analogously for $i = 1$.

\medskip
\noindent
For the induction step, assume $i\geq 2$ and that the statement in the lemma
holds for $i-1$. 
It is easy to verify by the construction of $\O$ and by induction hypothesis
that the following claim holds.

\medskip
\noindent
{\em Claim A:} For every $n \geq 1$ the following equivalences hold:
\begin{eqnarray*}
(q_0,n)\models\varphi_\diamond\wedge\exists\X\varphi_{i-1}  & \quad\Longleftrightarrow \quad &
(q_2,n)\models\varphi_\diamond\wedge\exists\X\varphi_{i-1}    \quad\Longleftrightarrow \quad
2^{i-1} \text{ divides } n \\
(q_1,n)\models\varphi_\diamond\wedge\exists\X\varphi_{i-1}  & \quad\Longleftrightarrow \quad &
(q_3,n)\models\varphi_\diamond\wedge\exists\X\varphi_{i-1}    \quad\Longleftrightarrow \quad
2^{i-1} \text{ does not divide } n
\end{eqnarray*}
Using Claim A, one can easily show the following (recall that
$\mu_i = \exists (\varphi_\diamond\wedge
\exists\X\varphi_{i-1}) \U (q_0\wedge\neg\exists\X q_1)$):

\medskip
\noindent
{\em Claim B:} For every $n \geq 0$ the following equivalences hold:
\begin{eqnarray*}
(q_0,n)\models\mu_i  & \quad\Longleftrightarrow \quad & 
    2^{i-1} \text{ divides } n \text{ and }
    |\{n'\in[n]\mid 2^{i-1}\text{ divides }n'\}|\text{ is even} \\
(q_1,n)\models\mu_i  & \quad\Longleftrightarrow \quad & 
    2^{i-1} \text{ does not divide } n \text{ and }
    |\{n'\in[n]\mid 2^{i-1}\text{ divides }n'\}|\text{ is odd} \\
(q_2,n)\models\mu_i  & \quad\Longleftrightarrow \quad & 
    2^{i-1} \text{ divides } n \text{ and }
    |\{n'\in[n]\mid 2^{i-1}\text{ divides }n'\}|\text{ is odd} \\
(q_3,n)\models\mu_i  & \quad\Longleftrightarrow \quad & 
    2^{i-1} \text{ does not divide } n \text{ and }
    |\{n'\in[n]\mid 2^{i-1}\text{ divides }n'\}|\text{ is even} \\
\end{eqnarray*}
Let us now prove Point (1) from the lemma for $i \geq 2$. 
We have the following equivalences:
\begin{eqnarray*}
(t,n)\models\varphi_i &\quad\Longleftrightarrow\quad &
(q_0,n)\models \mu_i \\
& \stackrel{\text{Claim B}}{\Longleftrightarrow}&
2^{i-1}\text{ divides } n \text{ and }
|\{n'\in[n]\mid 2^{i-1}\text{ divides }n'\}|\text{ is even}\\
& \stackrel{\text{Fact~\ref{F div}}}{\Longleftrightarrow}&
2^i \text{ divides } n
\end{eqnarray*}
For Point (2), we have the following equivalences:
\begin{eqnarray*}
(\overline{t},n)\models\varphi_i &\quad\Longleftrightarrow\quad &
\exists j\in\{1,2,3\}:
(q_j,n)\models \mu_i \\
& \stackrel{\text{Claim B}}{\Longleftrightarrow}&
\text{either }2^{i-1}\text{ does not divide } n \text{ and }
|\{n'\in[n]\mid 2^{i-1}\text{ divides }n'\}| \text{ is odd (i.e. }j=1\text{),}\\
&& \text{or } 2^{i-1}\text{ does not divide } n \text{ and }
|\{n'\in[n]\mid 2^{i-1}\text{ divides }n'\}| \text{ is even (i.e. }j=3\text{),}\\
&& \text{or } 2^{i-1}\text{ divides } n \text{ and }
|\{n'\in[n]\mid 2^{i-1}\text{ divides }n'\}| \text{ is odd (i.e. }j=2\text{)}\\
& \Longleftrightarrow&
2^{i-1}\text{ does not divide } n \text{ or }
(2^{i-1}\text{ divides }n \text{ and }
|\{n'\in[n]\mid 2^{i-1}\text{ divides }n'\}| \text{ is odd})\\
& \stackrel{\text{Fact~\ref{F div}}}{\Longleftrightarrow}&
2^i\text{ does not divide } n
\end{eqnarray*}
\qed
\end{proof}
For expressing if the $i^{\text{th}}$ bit of a natural is set to $1$, we make use of the
following fact.

\begin{fact}{\label{F bit}}
Let $n\geq 0$ and $i\geq 1$. Then 
$\bit_i(n)=1$ if and only if 
$|\{n'\in[n]\mid 2^{i-1} \text{ divides } n'\}|\ \text{ is odd.}$
\end{fact}
\begin{proof}
We have
\begin{eqnarray*}
\bit_i(n)=1\quad & \Longleftrightarrow & \quad n \text{ mod } 2^i\in[2^{i-1},2^i-1] \\
\quad & \Longleftrightarrow & \quad \exists r \in [0,2^{i-1}-1], k \geq 0 : n
= r + (2k+1) \cdot 2^{i-1} \\
& \Longleftrightarrow & \quad |\{n'\in[n]\mid 2^{i-1}\text{ divides }n'\}|\text { is odd}.
\end{eqnarray*}
\qed
\end{proof}
Let us now define a family of $\CTL$ formulas $(\psi_i)_{i\geq 1}$
such that for each $n\in\N$ we have
$\bit_i(n)=1$	if and only if
$(\overline{t},n)\models\psi_i$. We set
\begin{eqnarray*}
\psi_1 \quad &=&\quad \varphi_1\quad \text{ and } \\
\psi_i\quad &=&\quad  
\overline{t}\ \wedge\ \exists\X \left((q_1\vee q_2)\ \wedge\ \mu_i \right) 
\quad\text{for each }i>1.
\end{eqnarray*}
Fact~\ref{F size varphi} and the construction of $\psi_i$ immediately yield the
following fact.

\begin{fact}{\label{F size psi}}
$|\psi_i|\in O(i)$.
\end{fact}
The following lemma shows the correctness of the construction.

\begin{lemma}{\label{L Bit}}
Let $n \geq 0$ and let $i\geq 1$. Then
$(\overline{t},n)\models\psi_i$ if and only if $\bit_i(n)=1$.
\end{lemma}

\begin{proof}
The case $i=1$ is covered by Lemma~\ref{L Correctness}.
For $i \geq 2$, the following equivalences hold:
\begin{eqnarray*}
(\overline{t},n)\models\psi_i 
&\quad\Longleftrightarrow\quad &
(q_1,n)\models\mu_i\text{ or }
(q_2,n)\models\mu_i\\
&\stackrel{\text{Claim B}}{\Longleftrightarrow}&
\text{either }2^{i-1}\text{ does not divide }n\text{ and }
|\{n'\in[n]\mid 2^{i-1}\text{ divides } n'\}| \text{ is odd}\\
&&\text{or } 2^{i-1} \text{ divides }n \text{ and }
|\{n'\in[n]\mid 2^{i-1}\text{ divides } n'\}| \text{ is odd}\\
&\Longleftrightarrow&
|\{n'\in[n]\mid 2^{i-1}\text{ divides } n'\}| \text{ is odd}\\
&\stackrel{\text{Fact~\ref{F bit}}}{\Longleftrightarrow}&
\bit_i(n)=1
\end{eqnarray*}
\qed
\end{proof}
For our final step, let us give a reduction from QBF.
Let $\alpha$ be the following quantified Boolean formula
$$
\alpha\quad =\quad Q_kx_k\ \ Q_{k-1}x_{k-1}\ \ \cdots\ \ Q_1x_1\ \
\beta(x_1,\ldots,x_k),
$$
where $\beta$ is a Boolean formula over variables
$\{x_1,\ldots,x_k\}$ and $Q_i\in\{\exists,\forall\}$ is a quantifier for
each $i\in[k]$.
Our overall goal is to give a $\CTL$ formula $\theta$ such that
our QBF formula $\alpha$ is valid if and only if $(\overline{t},0)\models\theta$.
A truth assignment
$\vartheta:\{x_1,\ldots,x_k\}\rightarrow\{0,1\}$ corresponds to the
natural number $n(\vartheta)\in[0,2^k-1]$, where
$\bit_i(n(\vartheta))=1$ if and only if $\vartheta(x_i)=1$, for each $i\in[k]$.
First, let $\widehat{\beta}$ be the $\CTL$ formula that is obtained from the
Boolean formula $\beta$ by replacing every occurrence of every variable
$x_i$ by $\psi_i$.
Hence we obtain that for each $\vartheta:\{x_1,\ldots,x_{k}\}\rightarrow\{0,1\}$ 
we have $\vartheta\models\beta$ if and only if
$(\overline{t},n(\vartheta))\models\widehat{\beta}$ by Lemma~\ref{L Bit}.

It remains to define $\theta$. 
Recall that $\theta$ will be evaluated in $(\overline{t},0)$.
Let us parse our quantified Boolean formula $\alpha$ from left to right. 
Setting the variable $x_k$ to $1$ will correspond to adding $2^{k-1}$ to the
counter and getting to state $(\overline{t},2^{k-1})$. Setting $x_k$ to $0$ on the other
hand will correspond to adding $0$ to the counter and hence remaining in
state $(\overline{t},0)$. Next, setting $x_{k-1}$ to $1$ corresponds to adding to the
current counter value $2^{k-2}$, whereas setting $x_{k-1}$ to $0$ corresponds to
adding $0$, as expected. 
Adding zero to the counter will be realized by the finite path 
that jumps  from control location $\overline{t}$ to
$p_0$ and then back to $\overline{t}$.
Adding $2^{i-1}$ to the counter, on the other hand, will be realized by a
finite path that jumps from control
location $\overline{t}$ to $p_1$ (and thereby adds $1$ to the counter), then 
loops at $p_1$ as long as the counter value is not divisible by $2^{i-1}$
(which can be ensured by checking if $(p_1,n)\models \exists\X(\overline{t} \wedge \varphi_{i-1})$ by
Lemma~\ref{L Correctness}) and finally jumps back to $\overline{t}$ 
when the counter value is divisible by $2^{i-1}$ for the first time again.
We repeat this process until we have to set $x_1$
either to $1$ or to $0$. Eventually setting $x_1$ to $1$ will correspond to go
from $\overline{t}$ to
$p_1$ (hence adding $1$ to the counter) and then getting back to $\overline{t}$, whereas
setting $x_1$ to $0$ will correspond to go from $\overline{t}$ to $p_0$ and then back to
$\overline{t}$. After that, we finally check if $\widehat{\beta}$ holds. 
Recall that $Q_k,\ldots,Q_1$ are the quantifiers of our quantified Boolean
formula $\alpha$. For each $i\in[2,k]$, let us define formula $\theta_i$ as
\begin{eqnarray*}
\theta_i\quad &=& \quad Q_i\X\ 
\left((p_0\vee p_1)\bigcirc_i
\exists\left((p_0\vee
\exists\X(\overline{t} \wedge \varphi_{i-1}))\
\U\ (\overline{t}\wedge \neg\varphi_{i-1}\wedge\theta_{i-1})
)\biggl.\right)\right) \text{ and } \\
\theta_1\quad &=& \quad Q_1\X \left((p_0\vee p_1)\bigcirc_1 \exists\X\
\widehat{\beta}\right)
\end{eqnarray*}
with $\bigcirc_i=\wedge$ in case $Q_i=\exists$ and $\bigcirc_i=\,\rightarrow$ in
case $Q_i=\forall$ for each $i\in[k]$.
As expected, we put $\theta=\theta_k$. Observe that the size of $\theta$ is
polynomial in the size of $\alpha$ and that $\theta$ can be computed in logarithmic space
from $\alpha$. 
We finally obtain the following easy equivalence.
\begin{lemma}
The formula $\alpha$ is valid if and only if $(\overline{t},0)\in\sem{\theta}_{T(\O)}$.
\end{lemma}
This finishes our $\PSPACE$ lower bound proof for expression complexity of
$\CTL$ over 
one-counter nets. We have the following theorem.

\begin{theorem} \label{Theo-CTL-expression}
$\CTL$ model checking of the fixed one-counter net $\O$ from Figure~\ref{Fig-fixed-CCP} is
$\PSPACE$-hard.
\end{theorem}
Note that the formula $\theta$ in our 
reduction necessarily has a leftward until
depth that depends on the size of $\alpha$.
By Corollary~\ref{CTL upper bound} this cannot be avoided unless $\P=\PSPACE$.
Observe that in order to express divisibility by powers of two, 
our $\CTL$ formulas $(\varphi_i)_{i\geq 0}$ 
have a linearly growing leftward until depth.

\section{Tools from complexity theory}  \label{S Tools}

For Section~\ref{S Data}--\ref{S Markov} we need some concepts
from complexity theory.
The {\em $i^{\text{th}}$ level $\BH_i$ of the Boolean hierarchy}  is defined as follows: $\BH_1=\NP$, 
$\BH_{2i}=\{L_1\cap L_2\mid L_1\in\BH_{2i-1}, L_2\in\coNP \}$, and
$\BH_{2i+1}=\{L_1\cup L_2\mid L_1\in\BH_{2i}, L_2\in\NP \}$. 
The {\em Boolean hierarchy $\BH$} is defined as $\cup_{i\geq 1}\BH_i$.
The class $\mathsf{P}^{\mathsf{NP}}$ is the class of all problems 
that can be solved on a polynomially time bounded deterministic Turing machine with
access to an oracle from $\NP$.
By $\P^{\NP[\log]}$ we denote the class of all problems that can be solved on a
polynomially time bounded deterministic Turing machines which can have access to
an $\NP$-oracle only logarithmically many times.
It is known that $\BH \subseteq \P^{\NP[\log]}$.

For naturals $m \geq 1$ and $0 \leq M \leq 2^m-1$ 
let $\BIN_m(M) = \bit_1(M) \cdots \bit_m(M) \in \{0,1\}^m$ denote the $m$-bit binary representation of $M$.
In \cite{Wag87}, it was shown that the following problem is 
complete for $\mathsf{P}^{\mathsf{NP}}$:

\noindent
INPUT: A Boolean formula $\psi(x_1,\ldots,x_m)$?

\noindent
QUESTION: Is $\psi$ satisfiable and is the maximal number $M \in [0,2^m-1]$ with
$\psi(\BIN_m(M))=1$ even (i.e. is the  lexicographically maximal
satisfying assignment even)?

\subsection{Circuit complexity} \label{Sec-circuits}

More details on circuit complexity can be found in \cite{Vol99}.
A Boolean circuit $C = C(x_1,\ldots,x_n)$ 
is a directed acyclic graph (dag) with the following properties
(in the following, nodes of $C$ are called
{\em gates}, the in-degree (resp. out-degree) of a gate is called
its \emph{fan-in}  (resp. \emph{fan-out})):
\begin{itemize}
\item The gates with fan-in $0$ (they are called \emph{input gates} in the following)
are labeled with one of the symbols
$x_1$,  $\neg x_1, \ldots, x_n$,  $\neg x_n$.
\item Every gate with fan-in at least one is labeled with either AND or with OR.
\item The gates of fan-out 0 (they are called \emph{output gates} in the following)
are linearly ordered, we denote this order by $o_1, \ldots, o_m$ in the following.
\end{itemize}
Such a circuit computes a function $f_C : \{0,1\}^n \to \{0,1\}^m$ in the obvious
way. {\em Threshold circuits} may in addition to Boolean circuits contain 
{\em majority gates}. Such a gate outputs $1$ if and only if at least 
half of its input gates evaluate to $1$. 
The \emph{fan-in of a circuit} is the maximal fan-in of a gate in the circuit.
The \emph{size of a circuit} is the number of gates in the circuit.
The \emph{depth of a circuit} is the number of gates along a longest path
from an input gate to an output gate.
An \emph{$\AC^0$-circuit family} (resp. \emph{$\TC^0$-circuit family})
is a sequence $(C_n)_{n \geq 1}$ of Boolean circuits (resp. threshold
circuits) such that for some polynomial $p(n)$ and constant $c$:
\begin{itemize}
\item the size of $C_n$ is at most $p(n)$, 
\item the depth of $C_n$ is at most $c$, and
\item for each $m$ there is at most one circuit in $(C_n)_{n \geq 1}$
with exactly $m$ input gates.
\end{itemize}
An \emph{$\NC^1$-circuit family} 
is a sequence $(C_n)_{n \geq 1}$ of Boolean circuits such 
that for some polynomial $p(n)$ and constant $c$:
\begin{itemize}
\item the size of $C_n$ is at most $p(n)$,
\item the depth of $C_n$ is at most $c \cdot \log(n)$,
\item the fan-in of $C_n$ is at most $2$, and
\item for each $m$ there is at most one circuit in $(C_n)_{n \geq 1}$
with exactly $m$ input gates.
\end{itemize}
Circuit families of these types compute 
partial mappings on $\{0,1\}^*$ in the 
obvious way.\footnote{Note that we do not require to have for every $n \geq 0$ 
a circuit with exactly $n$ input gates in the family, therefore
the computed mapping is in general only partially defined.} 
Finally, a circuit family $(C_n)_{n \geq 0}$ is called
\emph{logspace-uniform} if there exists a logspace transducer
that computes on input $1^n$ a representation (e.g. as a 
node-labeled dag) of the circuit $C_n$. In the literature 
on circuit complexity one can find more restrictive notions
of uniformity, see e.g. \cite{Vol99}, but logspace uniformity 
suffices for our purposes. In fact, polynomial time uniformity
suffices for proving our lower bounds w.r.t. 
polynomial time reductions.

For our lower bound on the data complexity of $\CTL$, 
we use a deep result from \cite{ChDaLi01,HeAlBa02}. First, we need
a few definitions. Let $p_i$ denote the $i^{\text{th}}$ prime number.
It is well-known from number theory that the $i^{\text{th}}$ prime 
requires $O(\log(i))$ bits in its binary representation.
For a number $0 \leq M < \prod_{i=1}^m p_i$ we define the 
{\em Chinese remainder representation} $\CRR_m(M)$ as the 
Boolean tuple 
$$
\CRR_m(M) = (x_{i,r})_{i \in [m], 0 \leq r < p_i}\quad \text{ with } \\
x_{i,r} = \begin{cases} 1 & \text{if } M \text{ mod } p_i = r \\
0 & \text{else } 
\end{cases}
$$
By the following theorem, we can transform a CRR-representation very
efficiently into binary representation.

\begin{theorem}[{\cite[Thm.~3.3]{ChDaLi01}}] \label{theorem hesse und co}
There is a logspace-uniform $\NC^1$-circuit family 
$(B_m( (x_{i,r})_{i \in [m], 0 \leq r < p_i} ))_{m \geq 1}$  such that 
for every $m \geq 1$, $B_m$ has has $m$ output gates and
$$
\forall\, 0 \leq M < \prod_{i=1}^m p_i : B_m( \CRR_m(M) ) = \BIN_m(M \text{ mod
} 2^m) .
$$ 
\end{theorem}
By \cite{HeAlBa02}, we could replace logspace-uniform $\NC^1$-circuits
in Theorem~\ref{theorem hesse und co} even by $\mathsf{DLOGTIME}$-uniform
$\TC^0$-circuits. The existence of a $\mathsf{P}$-uniform $\NC^1$-circuit family
for converting from CRR-representation to binary representation
was already shown in \cite{BCH86}.

Usually the Chinese remainder representation of $M$ is the tuple $(r_i)_{i \in [m]}$,
where $r_i = M  \text{ mod } p_i$. Since the primes $p_i$ will be always given in unary
notation, there is no essential difference between this representation and our 
Chinese remainder representation. The latter is more suitable for our purpose.

\subsection{Serializability} \label{sec real}

Intuitively, a complexity class $\mathcal{C}_1$ is 
called $\mathcal{C}_2$-serializable (where $\mathcal{C}_2$
is another complexity class) if every language  
$L \in \mathcal{C}_1$ can be accepted in the following way: 
There exists a polynomial $p(n)$ and a $\mathcal{C}_2$-machine
(or $\mathcal{C}_2$-circuit family) $A$
such that $x \in L$ is checked in $2^{p(|x|)}$ many
stages, which are indexed 
by the strings from $\{0,1\}^{p(|x|)}$.
In stage $y \in \{0,1\}^{p(|x|)}$, $A$ 
gets from the stage indexed 
by the lexicographic predecessor of $y$ a constant number of 
bits $b_1, \ldots, b_c$ and computes from these bits, the 
index $y$ and the original input $x$ new bits $b'_1, \ldots, b'_c$
which are delivered to the lexicographic next stage. 
In \cite{CaFu91} it was shown that $\PSPACE$ is $\mathsf{P}$-serializable;
in \cite{HLSVW93} this result was sharpened to $\AC^0$-serializability,
see also \cite{Vol98}. It is not stated in \cite{HLSVW93,Vol98} but easy to 
see from the proofs that \emph{logspace-uniform} $\AC^0$ suffices for serializing
$\PSPACE$, see the appendix for more details.

For our purpose, a slightly different definition of 
$\AC^0$-serializability is useful: A language $L$ is $\AC^0$-serializable
if there exists an NFA $A$ over the alphabet
$\{0,1\}$, a polynomial $p(n)$, and a logspace-uniform
$\AC^0$-circuit family $(C_n)_{n \geq 0}$, where
$C_n$ has exactly $n+p(n)$ many inputs and one output, such that for every 
$x \in \{0,1\}^n$ we have:
$$
x \in L \ \Longleftrightarrow \ C_n(x,0^{p(n)}) \cdots C_n(x,1^{p(n)}) \in L(A),
$$
where ``$\cdots$'' refers to the lexicographic order on $\{0,1\}^{p(n)}$.
A proof that every language in $\PSPACE$ is $\AC^0$-serializable
in this sense can be found in the appendix.

\section{Data complexity for $\CTL$ is hard for $\PSPACE$}  \label{S Data}

In this section, we prove that also the data complexity of $\CTL$ over
one-counter nets is hard for $\PSPACE$ and therefore $\PSPACE$-complete
by the known upper bounds for the modal $\mu$-calculus \cite{Serr06}.
Let us fix the set of propositions $\Pmc=\{\alpha,\beta,\gamma\}$ 
for this section.
In the following, w.l.o.g. we allow in $\delta_0$ (resp. in $\delta_{>0}$) 
transitions of the kind $(q,k,q')$, where $k\in\N$ (resp. $k\in\Z$) is given
in unary representation with the expected intuitive meaning. 

\begin{proposition} \label{prop main}
For the fixed $\EF$ formula $\varphi = (\alpha  \to \exists\mathsf{X} (\beta
\wedge \EF ( \neg \exists \mathsf{X} \gamma)))$ 
the following problem can be solved with a logspace transducer:

\noindent
INPUT: A list of the first $m$ consecutive (unary encoded) prime numbers and
a Boolean formula $F = F((x_{i,r})_{i \in [m], 0 \leq r < p_i})$

\noindent
OUTPUT:  An OCN $\Omc(F)$ with distinguished control locations $\inp$ and $\out$,
such that for every number $0 \leq M < \prod_{i=1}^m p_i$ 
the following are equivalent:
\begin{itemize}
\item $F(\CRR_m(M))=1$
\item  There exists a $\sem{\varphi}_{T(\Omc(F))}$-path from $(\inp,M)$ to $(\out,M)$ in $T(\Omc(F))$.
\end{itemize}
\end{proposition}

\begin{proof}
W.l.o.g. we may assume that negations occur in $F$ only
in front of variables. Then, a negated formula $\neg x_{i,r}$ 
can be replaced by the disjunction $\bigvee \{ x_{i,k} \mid 0 \leq k < p_i, r
\neq k \}$. Note that this can be done in logspace, since the primes $p_i$ 
are given in unary. Hence, we can assume that $F$ does not contain
negations.

The idea is to traverse the Boolean formula $F$
with the  OCN $\Omc(F)$ in a depth first manner. Each time
a variable $x_{i,r}$ is  seen, the OCN may also enter another branch, where it is 
checked, whether the current counter value is congruent 
$r$ modulo $p_i$. Let
\begin{eqnarray*}
\Omc(F) & = & (Q, \{Q_\alpha, Q_\beta, Q_\gamma\}, \delta_0, \delta_{>0}),
\text{ where} \\
Q &=& \{ \inp(G), \out(G) \mid G \text{ is a subformula of } F \} \cup 
 \{ \div(p_1), \ldots,\div(p_m), \perp \} \\
Q_\alpha &=& \{ \inp(x_{i,r}) \mid i \in [m], 0 \leq r < p_i \} \\
Q_\beta &=& \{ \div(p_1), \ldots,\div(p_m) \} \\
Q_\gamma &=&  \{ \perp \}.
\end{eqnarray*}
We set $\inp = \inp(F)$ and $\out = \out(F)$.
Let us now define the transition sets $\delta_0$ and $\delta_{>0}$.
In case $G = G_1 \vee G_2$ is a subformula of $F$,
we add the following transitions to $\delta_0$ and $\delta_{>0}$:
\begin{gather*}
(\inp(G),0,\inp(G_i)), \ 
(\out(G_i),0, \out(G)) \text{ for $i \in\{1,2\}$.} \\
\end{gather*}
In case $G = G_1 \wedge G_2$ is a subformula of $F$,
we add the following transitions to $\delta_0$ and $\delta_{>0}$:
\begin{gather*}
(\inp(G),0,\inp(G_1)), \   
(\out(G_1),0,\inp(G_2)), \
(\out(G_2),0,\out(G)) . 
\end{gather*}
For every variable $x_{i,r}$
we add to $\delta_0$ and $\delta_{>0}$
the transition
$$
(\inp(x_{i,r}),0,\out(x_{i,r})) .
$$
Moreover, we add to $\delta_{>0}$ the transitions
$$
(\inp(x_{i,r}),-r,\div(p_i))  
$$
The transition $(\inp(x_{i,0}),0,\div(p_i))$ 
is also added to $\delta_0$.
For the control locations $\div(p_i)$ we add to $\delta_{>0}$ 
the transitions  $(\div(p_i),-p_i, \div(p_i))$ and
$(\div(p_i),-1,\perp)$.
This concludes the description of the OCN $\Omc(F)$.
Correctness of the construction can be easily checked by
induction on the structure of the formula $F$.
\qed
\end{proof}
We are now ready to prove $\PSPACE$-hardness of the data complexity.

\begin{theorem} \label{theo ctl data}
There exists a fixed $\CTL$ formula 
of the form $\exists \varphi_1 \U \varphi_2$, where
$\varphi_1$ and $\varphi_2$ are $\EF$ formulas, such that
the following problem is $\PSPACE$-complete:

\noindent
INPUT: An OCN $\Omc$ and a control location $q$ of $\Omc$.

\noindent
QUESTION: $(T(\Omc), (q,0)) \models \exists \varphi_1 \U \varphi_2$?
\end{theorem}

\begin{proof}
Let us take an arbitrary $\PSPACE$-complete language $L$. 
Recall from Section~\ref{sec real}
that $\PSPACE$ is $\AC^0$-serializable \cite{HLSVW93} and 
hence $\NC^1$-serializable.
Thus, there exists an NFA $A=(S,\{0,1\},\delta,s_0,S_f)$ over the alphabet
$\{0,1\}$, a polynomial $p(n)$, and a logspace-uniform
$\NC^1$-circuit family $(C_n)_{n \geq 0}$, where
$C_n$ has $n+p(n)$ many inputs, such that for every 
$x \in \{0,1\}^n$ we have:
\begin{equation} \label{eq C_n}
x \in L \ \Longleftrightarrow \ C_n(x,0^{p(n)}) \cdots C_n(x,1^{p(n)}) \in L(A),
\end{equation}
where ``$\cdots$'' refers to the lexicographic order on $\{0,1\}^{p(n)}$.
Fix an input $x \in \{0,1\}^n$. 
Our reduction can be split into the following five steps:

\medskip

\noindent
{\em Step 1.}
Construct in logarithmic space the circuit $C_n$.
Fix the the first $n$ inputs of $C_n$ to the bits in $x$,
and denote the resulting circuit by $C$; it has only
$m = p(n)$ many inputs. Equivalence (\ref{eq C_n}) can 
be written as
\begin{equation} \label{eq C}
x \in L \ \Longleftrightarrow \ \prod_{M=0}^{2^m-1} C(\BIN_m(M)) \in L(A).
\end{equation}
{\em Step 2.} Compute the first $m$ consecutive primes
$p_1, \ldots, p_m$. This is possible in logarithmic space, see e.g. 
\cite{ChDaLi01}. Note that every $p_i$ is bounded polynomially in $n$. 
Hence, every $p_i$ can be written down in unary notation.
Note that $\prod_{i=1}^m p_i > 2^m$ (if $m > 1$).

\medskip

\noindent
{\em Step 3.} Compute in logarithmic space
the circuit $B = B_m((x_{i,r})_{i \in [m], 0 \leq r < p_i} )$
from Theorem~\ref{theorem hesse und co}.
Thus, $B$ is a Boolean circuit of fan-in 2 and
depth $O(\log(m)) = O(\log(n))$ with
$$
B(\CRR_m(M)) = \BIN_m(M \text{ mod } 2^m)
$$ 
for every $0 \leq M < \prod_{i=1}^m p_i$.

\medskip

\noindent
{\em Step 4.}
Now we compose the circuits $B$ and $C$: For every $i\in[m]$, connect the $i^{\text{th}}$ input of the circuit
$C(x_1,\ldots,x_m)$ with the $i^{\text{th}}$ output of the circuit $B$.
The result is a circuit with fan-in 2 and depth $O(\log(n))$.
We can unfold this circuit into a Boolean formula 
$F = F((x_{i,r})_{i  \in [m], 0 \leq r < p_i})$. The resulting formula (or tree) has the same
depth as the circuit, i.e., depth $O(\log(n))$ and every tree node has at most
2 children. Hence, $F$ has polynomial size.
Thus, for every $0 \leq M < 2^m$ we have
$F(\CRR_m(M)) = C(\BIN_m(M))$ and equivalence (\ref{eq C}) can be written
as
\begin{equation} \label{gl F}
x \in L \ \Longleftrightarrow \ \prod_{M=0}^{2^m-1} F(\CRR_m(M)) \in L(A).
\end{equation}
{\em Step 5.} 
We now apply our construction from Proposition~\ref{prop main} to the formula $F$.
More precisely, let $G$ be the Boolean formula $\bigwedge_{i\in [m]} x_{i,r_i}$
were $r_i = 2^m \text{ mod } p_i$ for $i\in[m]$ (these remainders
can be computed in logarithmic space).
For every $1$-labeled transition $\tau \in \delta$ of the NFA $A$ let $\Omc(\tau)$
be a copy of the OCN $\Omc(F \wedge \neg G)$. 
For every $0$-labeled transition $\tau \in \delta$ let
$\Omc(\tau)$ be a copy of the OCN $\Omc(\neg F\wedge \neg G)$.
In both cases we write $\Omc(\tau)$
as $(Q(\tau),\{Q_\alpha(\tau),Q_\beta(\tau),Q_\gamma(\tau)\}, \delta_0(\tau),\delta_{>0}(\tau))$.
Denote with $\inp(\tau)$ 
(resp. $\out(\tau)$) the control location of this copy that corresponds to
$\inp$ (resp. $\out$) in $\Omc(F)$. Hence, for every $b$-labeled transition
$\tau \in \delta$ ($b \in \{0,1\}$) and every $0 \leq M < \prod_{i=1}^m p_i$
there exists a $\sem{\varphi}_{T(\Omc(\tau))}$-path  
($\varphi$ is from Proposition~\ref{prop main}) from
$(\inp(\tau),M)$ to $(\out(\tau),M)$ 
if and only if $F(\CRR_m(M))=b$ and $M \neq 2^m$.

We now define an OCN $\Omc = (Q, \{Q_\alpha, Q_\beta,Q_\gamma\}, \delta_0, \delta_{>0})$ as follows:
We take the disjoint union of all the OCNs $\Omc(\tau)$ for $\tau \in \delta$.
Moreover, every state $s \in S$ of the automaton $A$ becomes
a control location of $\Omc$:
\begin{eqnarray*}
Q &=& S \cup \bigcup_{\tau \in \delta} Q(\tau) \\
Q_p &=& \bigcup_{\tau \in \delta} Q_p(\tau) \text{ for } p \in \{\alpha,\beta,\gamma\} \\
\end{eqnarray*}
We add to $\delta_0$ and $\delta_{>0}$ for
every $\tau = (s,b,t) \in \delta$ the following transitions:
$$
(s,0,\inp(\tau)), \ (\out(\tau),+1,t) .
$$
Then, by Proposition~\ref{prop main} and (\ref{gl F}) we have
$x \in L$ if and only if there exists a 
$\sem{\varphi}_{T(\Omc)}$-path in $T(\Omc)$ from
$(s_0,0)$ to $(s,2^m)$ for some $s \in S_f$. 
Also note that there is no $\sem{\varphi}_{T(\Omc)}$-path in $T(\Omc)$
from $(s_0,0)$ to some configuration $(s,M)$ with $s \in S$ and 
$M > 2^m$. 
It remains to add to $\Omc$ some structure that enables
$\Omc$ to check that the counter has reached the value $2^m$.

For this, use Proposition~\ref{prop main} to construct the OCN $\Omc(G)$ (where $G$ is from above) 
and add it disjointly to $\Omc$. Moreover, add
to $\delta_{>0}$ and $\delta_0$ the transitions 
$(s,0,\inp)$ for all $s \in S_f$, where $\inp$ is the in control location  of $\Omc(G)$.
Finally, introduce a new proposition $\rho$ and set
$Q_\rho = \{\out\}$, where $\out$ is the out control location of $\Omc(G)$.
By putting $q=s_0$ we obtain:
$$
x \in L \ \Longleftrightarrow \ (T(\Omc), (q,0)) \models \exists \ 
\underbrace{(\alpha  \to \exists\mathsf{X} (\beta
\wedge \EF ( \neg \exists \mathsf{X} 
\gamma)))}_{\varphi\text{ from Proposition~\ref{prop main}}} \ \mathsf{U} \ \rho .
$$
This concludes the proof of the theorem.
\qed
\end{proof}
By slightly modifying the proof of Theorem~\ref{theo ctl data}, the following
corollary can be shown.

\begin{corollary}
There exists a fixed $\CTL$ formula of the kind $\exists\G\psi$, where $\psi$ is an
$\EF$ formula, such that the following problem is $\PSPACE$-complete:

\noindent
INPUT: An OCN $\O$ and a control location $q$ of $\O$.

\noindent
QUESTION: $(T(\O),(q,0))\models\exists\G\psi$?
\end{corollary}

\begin{proof}
The proof is almost identical to the proof of Theorem~\ref{theo ctl data},
except that we do not introduce the atomic proposition $\rho$. 
We rather add both to $\delta_0$ and $\delta_{>0}$ the transition
$(\out,0,\inp)$, where $\out$ is the out control location of 
$\Omc(G)$ and $\inp$ is the in control location of $\O(G)$.
We define $\psi=\exists\G\varphi$, where again $\varphi$ is
the formula from Proposition~\ref{prop main}.
\qed
\end{proof}

\section{Combined complexity of $\EF$ is hard for $\mathsf{P}^{\mathsf{NP}}$}
{\label{S Combined}}

In this section, we will apply the efficient transformation from
Chinese remainder representation to binary representation
(Theorem~\ref{theorem hesse und co}) in order
to prove that the combined complexity for 
$\EF$ over one-counter nets is hard for $\mathsf{P}^{\mathsf{NP}}$. 
For formulas represented succinctly by dags (directed acyclic graphs)
this was already shown in \cite{GoMaTo09}. The point here is that we use
the standard tree representation for formulas. 

\begin{proposition} \label{prop main ef}
The following problem can be solved by a logspace transducer:

\noindent
INPUT: A list of the first $m$ consecutive (unary encoded) prime numbers and
a Boolean circuit $C = C((x_{i,r})_{i \in [m], 0 \leq r < p_i})$ (with a
single output gate)

\noindent
OUTPUT:  An OCN $\Omc(C)$ with a distinguished state $\inp$
and an $\EF$ formula $\varphi(C)$ 
such that for every number $0 \leq M < \prod_{i=1}^m p_i$ 
we have:
$$
C(\CRR_m(M))=1 \quad \Longleftrightarrow \quad (T(\Omc(C)), (\inp,M)) \models \varphi(C) .
$$
\end{proposition}

\begin{proof}
As in the proof of Proposition~\ref{prop main} we can eliminate in $C$ all
input gates labeled with a negated variable.
Moreover, we can w.l.o.g. assume that the circuit $C$ is organized in $k+1$ layers,
where each layer either contains only AND- or OR-gates. All children
of a node in layer $i$ belong to layer $i+1$. Layer $1$ contains only
the unique output gate of the circuit, whereas layer $k+1$ contains the input gates.
For $i \in [k]$, let $\ell_i = \AND$ (resp. $\ell_i = \OR$) if layer
$i$ consists of AND-gates (resp. OR-gates).

The state set of the OCN $\Omc(C,b)$ contains all gates of the circuit
$C$; the unique output gate becomes the  distinguished state $\inp$.
We add the transition $(g_1, 0, g_2)$ to $\delta_0$ and $\delta_{>0}$
if gate $g_2$ is a child of gate $g_1$. 
If gate $g$ is an input gate labeled with $x_{i,r}$ then
we add the transition $(g, -r, \div(p_i))$ to $\delta_{>0}$.
If $r=0$, then the transition $(g, 0, \div(p_i))$ is also added
to $\delta_0$. 
Finally, for the states $\div(p_i)$ we have
the same transitions as in the proof of Proposition~\ref{prop main}.
This concludes the description of the OCN $\Omc(C)$.

In order to describe the $\EF$ formula $\varphi(C)$ 
let $\mathsf{M}_i = \exists\X$ (resp. $\mathsf{M}_i = \forall\X$) if 
$\ell_i = \OR$ (resp. $\ell_i = \AND$) for $i \in [k]$. Then let
\begin{equation}  \label{varphi(C,b)}
\varphi(C) = \mathsf{M}_1 \mathsf{M}_2 \cdots \mathsf{M}_k \exists\X \EF( \neg \exists \mathsf{X} \gamma) ,
\end{equation}
where the proposition $\gamma$ is used in the same way as in the proof
of Proposition~\ref{prop main} to allow to test if the counter value is zero.
It is clear that this formula fulfills the requirements of the theorem.
\qed
\end{proof}

\begin{theorem}{\label{T EF}}
The following problem is $\mathsf{P}^{\mathsf{NP}}$-hard:

\noindent
INPUT: An OCN $\Omc$, a state $q_0$ of $\Omc$, and an $\EF$ formula $\varphi$.

\noindent
QUESTION: $(T(\Omc), (q_0,0)) \models \varphi$?
\end{theorem}

\begin{proof}
Let us take a Boolean formula $\psi(x_1,\ldots,x_m)$.
We construct an OCN $\Omc_\psi$ with a distinguished state
$q_0$ and an $\EF$ formula $\varphi_\psi$
such that $(T(\Omc_\psi), (q_0,0)) \models \varphi_\psi$
if and only if $\psi$ is satisfiable and the maximal number $M \in [0,2^m-1]$ with
$\psi(\BIN_m(M))=1$ is even.

As in the proof of Theorem~\ref{theo ctl data} (Steps 2 and 3), we 
compute in logarithmic space the list $p_1, \ldots, p_m$ of 
the first $m$ consecutive primes and 
the circuit $B = B_m((x_{i,r})_{i \in [m], 0 \leq r < p_i} )$
of logarithmic depth and fan-in at most two
from Theorem~\ref{theorem hesse und co}.
We combine $B$ with the  Boolean formula $\psi(x_1,\ldots,x_m)$
and obtain a Boolean circuit $C =  C((x_{i,r})_{i \in [m], 0 \leq r < p_i})$
such that for every number $0 \leq M \leq 2^m-1$:
\begin{equation} \label{gl varphi 1}
\psi(\BIN_m(M)) = 1 \ \Longleftrightarrow \  C(\CRR_m(M)) = 1.
\end{equation}
As in the proof of Theorem~\ref{theo ctl data}
let $G$ be the Boolean formula $\bigwedge_{i\in [m]} x_{i,r_i}$
were $r_i = 2^m \text{ mod } p_i$ for $i\in[m]$.

The main structure of the OCN $\Omc_\psi$ is described by the 
following diagram:
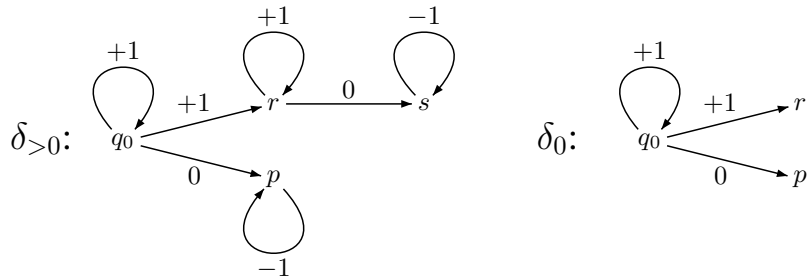
\begin{center}
\setlength{\unitlength}{1mm}
\begin{picture}(100,40)(-10,-15)
\gasset{Nadjustdist=0.8,Nadjust=wh,Nframe=n,Nfill=n,AHnb=1,ELdist=.8,linewidth=.2} 
  \put(-15,4){\Large $\delta_{>0}$:}
  \node(q)(0,5){$q_0$}
  \drawloop[loopangle=90](q){$+1$}
  \node(p)(20,0){$p$} 
  \drawedge[ELside=r](q,p){$0$} 
  \drawloop[loopangle=-90](p){$-1$}
  \node(r)(20,10){$r$} 
  \drawedge(q,r){$+1$} 
  \drawloop[loopangle=90](r){$+1$}
  \node(s)(40,10){$s$} 
  \drawedge(r,s){$0$} 
  \drawloop[loopangle=90](s){$-1$}
  \put(55,4){\Large $\delta_{0}$:}
  \node(q)(70,5){$q_0$}
  \drawloop[loopangle=90](q){$+1$}
  \node(p)(90,0){$p$} 
  \drawedge[ELside=r](q,p){$0$} 
  \node(r)(90,10){$r$} 
  \drawedge(q,r){$+1$} 
   \end{picture}
\end{center}
{}From the states $q_0$, $p, r$, and $s$ some further $0$-labeled transitions
emanate to OCNs of the form constructed in Proposition~\ref{prop main ef}:
\begin{itemize}
\item From $q_0$ a transition into the initial state $\inp$ of a copy of 
$\Omc(C)$.
\item From $p$ and $s$ a transition into the initial state $\inp$ of a copy of 
$\Omc(G)$.
\item From $r$ a transition into the initial state $\inp$ of a copy of 
$\Omc(\neg C)$.
\end{itemize}
Now our $\EF$ formula $\varphi_\psi$ expresses the following:
We can reach a configuration $(q_0,M_1)$ from $(q_0,0)$ in the OCN
$\Omc_\psi$ such that the following holds:
\begin{itemize}
\item $C(\CRR_m(M_1)) = 1$,
\item from $(q_0,M_1)$ we cannot reach a configuration
$(p,M_0)$ with $0 \leq M_0 \leq M_1$ and $G(\CRR_m(M_0)) = 1$
(i.e., $M_0 = 2^m \text{ mod } \prod_{i=1}^m p_i$), and
\item for all configurations $(r,M_2)$  that are reachable from $(q_0,M_1)$
(hence  $M_2 > M_1$)
the following holds: If we cannot reach a configuration 
$(s, M_3)$ from $(r, M_2)$ with $G(\CRR_m(M_3)) = 1$ 
then $C(\CRR_m(M_2)) = 0$.
\end{itemize}
Using the formulas constructed in Proposition~\ref{prop main ef},
it is straightforward to transform this description into a real
$\EF$ formula. This concludes the proof.
\qed
\end{proof}
At the moment we cannot prove $\mathsf{P}^\mathsf{NP}$-hardness
for the data complexity of $\EF$ over OCPs.
For this, it would be sufficient to have a fixed $\EF$ formula 
$\varphi(C)$ in (\ref{varphi(C,b)}). Note that this formula
only depends on the number of layers $k$ of the circuit $C$. Hence, if 
$C$ is from an $\AC^0$-circuit family, then $\varphi(C)$
is in fact a fixed formula.  In our case, the circuit is the composition
of two circuits, one from an $\NC^1$-circuit family
(coming from Theorem~\ref{theorem hesse und co}, where we could
even assume a $\TC^0$-circuit family) 
and a Boolean formula, which can be assumed to be in conjunctive normal form.
Hence, the main obstacle for getting a fixed formula
is the fact that converting from Chinese remainder representation
to binary representation is not possible in $\AC^0$ (this is provably
the case).

\section{Reachability objectives on one-counter Markov decision processes} \label{S Markov}

In this section we show that the techniques developed in the previous sections
can be used to improve a lower bound  
on verifying reachability objectives on one-counter
Markov decision processes from \cite{BraBroEteKucWoj09}.

\renewcommand{\D}{\mathcal{D}}
\renewcommand{\C}{\mathcal{C}}
\newcommand{\Reach}{\text{Reach}}
\newcommand{\Val}{\text{Val}}

A {\em probability distribution} on a non-empty finite set $S$ is a
function $f:S\rightarrow\{x\in\Q\mid 0\leq x\leq 1\}$ such that
$\sum_{s\in S}f(s)=1$. We restrict here to rational probabilities,
in order to get finite representations for probability distributions.
A (image-finite) {\em Markov chain} is a triple $\C=(S,\rightarrow,f)$, where 
$(S,\to)$ is an image-finite and deadlock-free directed graph ($S$ is also called
the set of states of $\C$)  and
$f$ assigns to each $s\in S$ a probability distribution $f(s)$ over all (the finitely
many) successors of $s$ w.r.t. $\to$.  If $s \to t$, then we also use the notations
$f(s,t) = x$ or $s\xrightarrow{x}t$ for $(f(s))(t) = x \in \Q$.
A (image-finite) {\em Markov decision process} (MDP) is a triple
$\D=(V,\hookrightarrow,f)$, where 
$(V,\hookrightarrow)$ is again an image-finite and deadlock-free directed graph,
the set $V$ of vertices is partitioned as $V=V_N\uplus V_P$
($V_N$ is the set of  {\em nondeterministic} vertices, $V_P$ 
is the set of {\em probabilistic} vertices),
and $f$ assigns to each probabilistic vertex $v\in V_P$ a probability distribution on $v$'s successors.
A {\em strategy} $\sigma$ is a function that assigns to each $wv$ with 
$w\in V^*$ and $v\in V_N$ a probability distribution on $v$'s successors.
If  $\sigma$ assigns to $wv$ and $v'$ (where $v\hookrightarrow
v'$) the probability $x$, then we write $\sigma(wv,v')=x$.
Every strategy $\sigma$ determines a Markov chain
$\D(\sigma) = (V^+, \to, f)$, where $wv\xrightarrow{x}wvv'$ if and only if
$v\hookrightarrow v'$ and moreover either $v\in V_P$ and $f(v,v')=x$, or $v\in
V_N$ and $\sigma(wv,v')=x$.
Let $\path_\omega(\D) = \path_\omega(V,\hookrightarrow)$ and 
$\path_\omega(\D(\sigma)) = \path_\omega(V^+,\to)$; paths in these
sets will be called {\em runs} in $\D$ or $\D(\sigma)$, respectively.
Note that every run in $\D$ corresponds to a unique run in $\D(\sigma)$
and vice versa in a natural way. In order to simplify notation, we will quite often
identify these corresponding runs.
Let us fix a set of {\em target vertices} (also called a {\em reachability objective})
$T\subseteq V$ of the MDP $\D$. 
For each strategy $\sigma$ and each vertex $v\in V$ of $\D$, let
$$
\Reach_T^\sigma(v)=\{w\in\path_\omega(\D(\sigma))\mid w_1=v\text{ and } \exists i \geq 1 : w_i \in V^*T\}
$$ 
denote all runs in $\D(\sigma)$ that start in $v$
and that satisfy the reachability objective $T$ in $\D$.
For each $T$ and each $v$, the set $\Reach_T^\sigma(v)$ is measurable.
The probability $\mathcal{P}(\Reach_T^\sigma(v))$ for the set $\Reach_T^\sigma(v)$ can be obtained as follows:
Take all finite paths $w \in \path_f(\D(\sigma))$ that start in $v$ and
such that the last state of $w$ is from $V^*T$ but no previous state in $w$ 
is from $V^*T$ (this set is prefix free). For each such finite path $w = w_1 \cdots w_n$ 
such that $w_i \xrightarrow{x_i} w_{i+1}$ in $\D(\sigma)$ the probability
is $x_1 \cdot x_2 \cdots x_{n-1}$. Finally, the probability for 
$\Reach_T^\sigma(v)$ is the (possibly infinite) sum of all these probabilities.
Now, let us define the {\em $T$-reachability value in $v$} by
$$\Reach_T(v)\ =\ \sup\{\mathcal{P}(\Reach_T^\sigma(v))\mid 
\sigma\text{ is a strategy in }\D\}.
$$ 
Observe that it is not required that this supremum is actually reached
by a certain strategy $\sigma$.
If however a strategy $\sigma$ does reach the $T$-reachability value, i.e.,
$\mathcal{P}(\Reach^\sigma_T(v))=\Reach_T(v)$, then $\sigma$ is called
{\em optimal}.

\renewcommand{\A}{\mathcal{A}}

A one-counter Markov decision process (OC-MDP) is a tuple
$\A=(Q,\delta_0,\delta_{>0},f_0,f_{>0})$, where $Q=Q_N\uplus Q_P$ is a finite
set of {\em control locations} which is partitioned into {\em nondeterministic
control locations} $Q_N$ and {\em probabilistic control locations} $Q_P$,
$\delta_0\subseteq Q\times\{0,1\}\times Q$ is  a set of {\em zero transitions}
and
$\delta_{>0}\subseteq Q\times\{-1,0,1\}\times Q$ is  a set of {\em positive transitions}
such that each $q\in Q$ has at least one outgoing zero transition and at least
one outgoing positive transition, and finally $f_0$ (resp. $f_{>0}$) assigns to 
each $q\in Q_P$ a probability distribution over all outgoing zero (resp. positive)
transitions of $q$. 
The MDP that $\A$ describes is $\D(\A)=(V,\hookrightarrow,f)$, where
\begin{itemize}
\item $V_N=Q_N\times\N$ and $V_P=Q_P\times\N$, and
\item $(q,n)\hookrightarrow(q',n+i)$ if and only if one of the following two
holds:
\begin{itemize}
\item $n=0$ and $(q,i,q')\in\delta_0$. In this case
$f$ assigns to $(q,n)\hookrightarrow(q',n+i)$ the probability $f_0(q,i,q')$.
\item $n>0$ and $(q,i,q')\in\delta_{>0}$.
In this case
$f$ assigns to $(q,n)\hookrightarrow(q',n+i)$ the probability $f_{>0}(q,i,q')$.
\end{itemize}
\end{itemize}
Given an OC-MDP $\A=(Q,\delta_0,\delta_{>0},f_0,f_{>0})$ and a set of 
control locations $R\subseteq Q$, define
$$
\ValOne(R)\ =\ \{(q,n)\in Q\times\N\mid \Reach_{R\times\{0\}}(q,n)=1\}
$$
and
$$\OptValOne(R)\ =\ \{(q,n)\in Q\times\N\mid \exists \text{ strategy }\sigma:
\mathcal{P}(\Reach^\sigma_{R\times\{0\}}(q,n))=1\}$$
(both sets are defined w.r.t $\D(\A)$). 
In other words: $\ValOne(R)$ is the set of all states $(q,n)$
of the MDP $\D(\A)$ such that for every $\epsilon > 0$ there exists
a strategy $\sigma_\epsilon$ under which the
probability of reaching from $(q,n)$ a control location in $R$ and at
the same time having counter value $0$ is at least $1-\varepsilon$.
$\OptValOne(R)$ is the set of all states $(q,n)$
of the MDP $\D(\A)$ for which there exists a specific strategy under
which this probability becomes $1$.

\begin{theorem}[\cite{BraBroEteKucWoj09}]  \label{OC-MDP-1}
The following problem is $\PSPACE$-hard and in $\EXPTIME$:

\noindent
INPUT: An OCP-MDP $\A=(Q,\delta_0,\delta_{>0},f_0,f_{>0})$, $R\subseteq Q$, and
$q\in Q$.

\noindent
QUESTION: $(q,0)\in\OptValOne(R)$?
\end{theorem}

\noindent
Theorem~\ref{OC-MDP-1} was proven by a reduction from the
$\PSPACE$-complete emptiness problem for alternating finite word automata
over a singleton alphabet (\cite{Ho96}, see also \cite{JaSa07} for a simplified
presentation). 

\begin{theorem}[\cite{BraBroEteKucWoj09}]  \label{OC-MDP-2}
The following problem is hard for every level of $\BH$:

\noindent
INPUT: An OC-MDP $\A=(Q,\delta_0,\delta_{>0},f_0,f_{>0})$, $R\subseteq Q$, and
$q\in Q$.

\noindent
QUESTION: $(q,0)\in\ValOne(R)$?
\end{theorem}
Currently, it is open whether the problem stated in Theorem~\ref{OC-MDP-2} is decidable;
the corresponding problem for 
MDPs defined by pushdown processes is undecidable \cite{EtYa05}.

From the proof of Theorem~\ref{OC-MDP-2} 
it can be seen that the authors prove actually hardness for
$\P^{\NP[\log]}$. Moreover, it is pointed out in  \cite{BraBroEteKucWoj09} that various difficulties arise
when trying to improve the latter lower bound.
In this section, we will improve the lower bound for membership in $\ValOne(R)$ 
to $\PSPACE$. From our proof one can easily see that we reprove
$\PSPACE$-hardness of $\OptValOne$ as a byproduct. But first, we need the
following lemma.

\begin{lemma} \label{L Prob}
The following problem can be solved by a logspace transducer:

\noindent
INPUT: A list of the first $m$ consecutive (encoded in unary) prime numbers and
a Boolean formula $F=F((x_{i,r})_{i \in [m], 0 \leq r < p_i})$.

\noindent
OUTPUT: An OC-MDP $\A=\A(F)$ with control locations $Q$, a set $R=R(F)\subseteq Q$, and
some control location $q_F\in Q$ such that for every number $0\leq
M<\prod_{i=1}^m p_i$ the following holds:
\begin{itemize}
\item If $F(\CRR_m(M))=1$, then there exists a strategy $\sigma$ such that
$\mathcal{P}(\Reach^\sigma_{R\times\{0\}}(q_F,M))=1$.
\item If $F(\CRR_m(M))=0$, then for every strategy $\sigma$ we have
$\mathcal{P}(\Reach^\sigma_{R\times\{0\}}(q_F,M))\leq 1-2^{-|F|}$.
\end{itemize}
\end{lemma}

\begin{proof}
As in the proof of Proposition~\ref{prop main} we can eliminate all
input gates labeled with a negated variable $\neg x_{i,r}$.
The OC-MDP $\A=\A(F)=(Q,\delta_0,\delta_{>0},f_0,f_{>0})$ will have for each 
subformula $G$ of $F$ a control
location $q_G$. If $G$ is of the form $G=G_1\vee G_2$, then
$q_G$ will be nondeterministic and both in $\delta_0$ and in $\delta_{>0}$
there is a transition from $q_G$ to both $q_{G_1}$ and
$q_{G_2}$ that does not change the counter value. 
If $G$ is of the form $G=G_1\wedge G_2$, then $q_G$ will be probabilistic
and both in $\delta_0$ and in $\delta_{>0}$ there will be a transition to
both $q_{G_1}$ and $q_{G_2}$ that does not
change the counter value and which  will be chosen with probability ${1 \over
2}$ each. 
Now assume that $G$ is a variable $x_{i,r}$. 
Recall that  $x_{i,r}$ is set to one if and only if $M \text{ mod } p_i=r$. We
introduce in $\A$ further (deterministically behaving) control locations 
$q(j,p_i)$ for $0 \leq j < p_i$ that allow to test if $M$ is
congruent $r$ modulo $p_i$ by allowing the following transitions in
$\delta_{>0}$ for each $0 \leq j < p_i$:
$$
(q(j,p_i),-1,q(j-1\text{ mod } p_i,p_i))
$$ 
Since each $q(j,p_i)$ has to have an outgoing transition both in $\delta_0$ and
$\delta_{>0}$, we add the transition $$(q(j,p_i),0,q(j,p_i))$$ to $\delta_0$
for each $0 \leq j < p_i$.
We put $q_{x_{i,r}}$ to be nondeterministic with a transition both in $\delta_0$ and in
$\delta_{>0}$ from $q_{x_{i,r}}$ to $q(r,p_i)$ that does not change the counter value.
Finally we put  $R=\{q(0,p_i)\mid i\in[m]\}$.

Assume first that $F(\CRR_m(M))=1$. We prove that there exists a strategy
$\sigma$ such that $$\mathcal{P}(\Reach^\sigma_{R\times\{0\}}(q_F,M))=1$$ 
in $\D(\A)$.  Note that the only  
nondeterministic states in
$\D(\A)$ that have more than one successor are states which correspond to a
disjunctive subformula $G=G_1\vee G_2$ of $F$. If
$G(\CRR_m(M))=1$, then there exists some $i\in\{1,2\}$ such that 
$G_i(\CRR_m(M))=1$. Our strategy $\sigma$ will choose $(q_G,M)$'s successor
$(q_{G_i},M)$ with probability $1$. If $G(\CRR_m(M))=0$, then the choice of
$\sigma$ is irrelevant and we let
$\sigma$ choose $(q_G,M)$'s successor uniformly distributed, say.
It is now easy to verify 
that $\mathcal{P}(\Reach^\sigma_{R\times\{0\}}(q_F,M))=1$.

On the other hand, assume that $F(\CRR_m(M))=0$ and consider an arbitrary
strategy $\sigma$. 
The question is how close can $\mathcal{P}(\Reach^\sigma_{R\times\{0\}}(q_F,M))$
reach $1$.  We prove by induction on the structure of the formula $F$ that
\begin{equation}  \label{case F=0}
\mathcal{P}(\Reach^\sigma_{R\times\{0\}}(q_F,M))\leq 1-2^{-k},
\end{equation}
where $k$ is the number of conjunctions that appear in $F$. 
If $F$ is a variable $x_{i,r}$, then 
$$
\mathcal{P}(\Reach^\sigma_{R\times\{0\}}(q_F,M)) = 0 = 1 -2^0.
$$
If $F = F_1 \vee F_2$ then $F_1(\CRR_m(M))=F_2(\CRR_m(M))=0$.
Assume that $\sigma$ assigns to the transition from $(q_F,M)$ to $(q_{F_i},M)$ 
the probability $x_i$, where $x_1+x_2 = 1$. With the induction hypothesis, we get
\begin{eqnarray*}
\mathcal{P}(\Reach^\sigma_{R\times\{0\}}(q_F,M))  & = & 
x_1 \cdot \mathcal{P}(\Reach^\sigma_{R\times\{0\}}(q_{F_1},M))+x_2 \cdot \mathcal{P}(\Reach^\sigma_{R\times\{0\}}(q_{F_2},M)) \\
& \leq &
x_1 (1-2^{-k_1} ) + x_2 (1-2^{-k_2} ), 
\end{eqnarray*}
where $k_i$ the number of conjunctions that appear in $F_i$. Since $k_i \leq k$, we get (\ref{case F=0}).
Finally, assume that $F = F_1 \wedge F_2$ and let $k_i$ be the number of conjunctions that appear in $F_i$.
Hence, $k_i \leq k-1$. If $F_1(\CRR_m(M))=F_2(\CRR_m(M))=0$ then we get 
$\mathcal{P}(\Reach^\sigma_{R\times\{0\}}(q_F,M)) \leq 1-2^{-k+1} \leq 1-2^{-k}$. 
On the other hand, if e.g. $F_1(\CRR_m(M))=0$ but $F_2(\CRR_m(M))=1$ (the other case
is symmetric), then we get
\begin{eqnarray*}
\mathcal{P}(\Reach^\sigma_{R\times\{0\}}(q_F,M)) & = & \frac{1}{2}  \cdot \mathcal{P}(\Reach^\sigma_{R\times\{0\}}(q_{F_1},M))+ 
\frac{1}{2} \cdot \mathcal{P}(\Reach^\sigma_{R\times\{0\}}(q_{F_2},M)) \\
&\leq  & \frac{1}{2}  \cdot (1-2^{-k+1}) + \frac{1}{2} = 1-2^{-k} .
\end{eqnarray*}
This concludes the proof of (\ref{case F=0}). Since $k\leq |F|$ we
obtain
$\mathcal{P}(\Reach^\sigma_{R\times\{0\}}(q_F,M))\leq 1-2^{-|F|}$.
This concludes the proof of Lemma~\ref{L Prob}.
\qed
\end{proof}

\begin{theorem}{\label{T ValOne}}
The following problem is $\PSPACE$-hard:

\noindent
INPUT: An OC-MDP $\A=(Q,\delta_0,\delta_{>0},f_0,f_{>0})$, $R\subseteq Q$, and $q\in Q$.

\noindent
QUESTION: $(q,0)\in\ValOne(R)$?
\end{theorem}
\begin{proof}
Let $L \subseteq \{0,1\}^*$ be an arbitrary $\PSPACE$-complete language, let $x\in\{0,1\}^*$ be a
word of length $n$.  
We repeat steps 1 to 4  of the proof of Theorem~\ref{theo ctl data}.
This means, we compute in logspace a Boolean formula
$F=F((x_{i,r})_{i \in [m], 0 \leq r < p_i})$ of polynomial
size in $n$ such that for some fixed NFA $A =(S,\{0,1\},\delta,s_0,S_f)$
we have
$$
x\in L\quad\Longleftrightarrow\quad \prod_{M=0}^{2^m-1}F(\CRR_m(M))\in L(A).
$$
By doubling, if necessary, the set of final states of $A$ we can
assume that states from $S_f$ do not have outgoing
transitions but every state from $S \setminus S_f$ has at least
one outgoing transition. This assumption
will slightly simplify our construction. 

Let $G=\bigwedge_{i\in[m]} x_{i,r_i}$ with $r_i=2^m \text{ mod } p_i$ for
each $i\in[m]$ be the Boolean formula that tests if $M$ equals $2^m$.
We will build an OC-MDP $\A=(Q,\delta_0,\delta_{>0},f_0,f_{>0})$ with $S \subseteq Q$ and a target set of
control locations $R\subseteq Q$ such that 
$$
\prod_{M=0}^{2^m-1}F(\CRR_m(M))\in L(A)\ 
\quad\Longleftrightarrow\quad \Reach_{R\times\{0\}}(s_0,0)=1.
$$
Moreover, our reduction will have the additional property that 
$$\Reach_{R\times\{0\}}(s_0, 0)=1\quad\Longleftrightarrow\quad\exists \sigma:
\mathcal{P}(\Reach^\sigma_{R\times\{0\}}(s_0, 0))=1.
$$
Hence, we prove $\PSPACE$-hardness of $\OptValOne$ as a byproduct.
The control locations in $S \setminus S_f$ are nondeterministic in $\A$ ($\A$ will
hence behave nondeterministically in control locations from $S \setminus S_f$).
The NFA
$A$ on input $F(\CRR_m(0)) \cdots F(\CRR_m(2^m-1))$ will be simulated by $\A$ from state
$(s_0, 0)$ by consecutively incrementing the counter and checking if for the
current  counter value $M$ and for the current (to be simulated) $b$-labeled transition of $A$
we have $F(\CRR_m(M))=b$. This simulation will be done
until a state $(s,2^m)$ with $s \in S_f$ is reached.
Recall that by Lemma~\ref{L Prob} we can compute  
OC-MDPs $\A(F \wedge \neg G)$,  $\A(\neg F \wedge \neg G)$, and $\A(G)$
together with sets of control locations $R(F \wedge \neg G)$, $R(\neg F \wedge \neg G)$,
and $R(G)$, and  control locations
$q_{F \wedge \neg G}$, $q_{\neg F \wedge \neg G}$,
and $q_G$ such that, e.g., $\A(F \wedge \neg G)$ satisfies for each $0\leq M<\prod_{i=1}^m p_i$:
\begin{eqnarray*}
F(\CRR_m(M))=1 \wedge M \neq 2^m  & \Rightarrow &  \exists \text{ strategy } \sigma :
\mathcal{P}(\Reach^\sigma_{R(F \wedge \neg G)\times\{0\}}(q_{F \wedge \neg G},M))=1 \\
F(\CRR_m(M))=0 \vee M = 2^m  & \Rightarrow  & \forall \text{ strategies } 
\sigma : \mathcal{P}(\Reach^\sigma_{R(F \wedge \neg G)\times\{0\}}(q_{F \wedge \neg G},M)) \leq 1-2^{-|F \wedge \neg G|}
\end{eqnarray*}
The OC-MDPs $\A(\neg F \wedge \neg G)$ and $\A(G)$ have analogous properties.

In the following diagrams we draw transitions that do {\em not} modify the counter value in normal
width and we draw transitions that increase the counter value by one in thicker
width. 
We realize each NFA-transition $(s,1,t)\in\delta$ with $s \not\in S_f$ both in $\delta_0$
and in $\delta_{>0}$ by
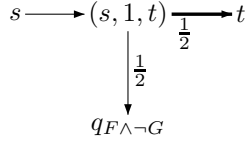
\begin{center}
\setlength{\unitlength}{0.05cm}
\begin{picture}(60,35)(0,0)
 \gasset{Nframe=n,loopdiam=9,ELdist=1}
 \gasset{Nadjust=wh,Nadjustdist=1}
 \gasset{curvedepth=0}
 \node(s)(0,30){$s$}
 \node(b)(30,30){$(s,1,t)$}
 \node(t)(60,30){$t$}
 \node(F)(30,0){$q_{F\wedge\neg G}$}
 \drawedge(s,b){}
 \drawedge[linewidth=.9,ELside=r](b,t){${1\over 2}$}
 \drawedge(b,F){${1\over 2}$}
\end{picture}
\end{center}
whereas each transition $(s,0,t)\in\delta$ with $s \not\in S_f$ is realized in $\A$ by
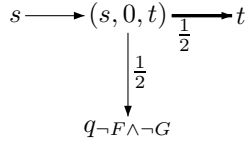
\begin{center}
\setlength{\unitlength}{0.05cm}
\begin{picture}(60,35)(0,0)
 \gasset{Nframe=n,loopdiam=9,ELdist=1}
 \gasset{Nadjust=wh,Nadjustdist=1}
 \gasset{curvedepth=0}
 \node(s)(0,30){$s$}
 \node(b)(30,30){$(s,0,t)$}
 \node(t)(60,30){$t$}
 \node(F)(30,0){$q_{\neg F\wedge\neg G}$}
 \drawedge(s,b){}
 \drawedge[linewidth=.9,ELside=r](b,t){${1\over 2}$}
 \drawedge(b,F){${1\over 2}$}
\end{picture}
\end{center}
i.e. we connect the intermediate control location $(s,b,t) \in \delta$ to  
$\A(F\wedge\neg G)$ (if $b=1$) or $\A(\neg F\wedge\neg G)$ (if $b=0$)
for checking if $F(\CRR_m(M))=b$ and $M<2^m$ for the
current counter value $M$. Moreover, for all final states $s \in S_f$ we add
a transition $s \xrightarrow{1} q_G$ to both $\delta_0$ and
$\delta_{>0}$ that does not change the counter value.
As expected, we put $R=R(F\wedge \neg G)\cup R(\neg F\wedge\neg G)\cup R(G)$.
Let $\D = \D(\A)$ in the following. Note that since every non-final state
has at least one outgoing transition in $A$,
$\D$ is indeed an MDP, i.e., the underlying graph is deadlock-free.

Now assume that $x\in L$. 
We show that there exists a strategy $\sigma$ such that 
$\mathcal{P}(\Reach^\sigma_{R\times\{0\}}(s_0, 0))=1$. 
Since $x\in L$, we have $\prod_{M=0}^{2^m-1}F(\CRR_m(M))\in L(A)$
along with some accepting run
$$
s_0\xrightarrow{b_0}s_1\xrightarrow{b_1} \quad \cdots\quad
s_{2^m-1}\xrightarrow{b_{2^m-1}}s_{2^m} \in S_f,
$$
where $s_M \not\in S_f$ and
$b_M=F(\CRR_m(M))$ for all $M\in[0,2^m-1]$.
For each $M\in[0,2^m-1]$ our strategy $\sigma$ will assign to $(s_M, M)$'s successor
$( (s_M,b_M,s_{M+1}), M)$ probability $1$. Moreover, by Lemma~\ref{L Prob} we can choose the 
strategy $\sigma$ such that:
\begin{eqnarray*}
b_M=1&\quad \Longrightarrow\quad&
\mathcal{P}(\Reach^\sigma_{R\times\{0\}}(q_{F\wedge\neg G},M))=1\\
b_M=0&\quad \Longrightarrow\quad&
\mathcal{P}(\Reach^\sigma_{R\times\{0\}}(q_{\neg F\wedge\neg G},M))=1
\end{eqnarray*}
for each $0\leq M <2^m$ and  
$\mathcal{P}(\Reach^\sigma_{R\times\{0\}}(q_G, 2^m))=1$.
It follows
$$\mathcal{P}(\Reach^\sigma_{R\times\{0\}}(s_0, 0))=1.$$
Conversely, assume now that $x\not\in L$. 
Our goal is to prove a global non-zero lower bound on the probability of runs in $\D(\sigma)$ that 
begin in $(s_0, 0)$ and that do {\em not} reach $R\times\{0\}$, where 
$\sigma$ is an arbitrary strategy.
For this, let us first fix an arbitrary strategy $\sigma$ in $\D$.
We distinguish the following three types $(A)$, $(B)$ and $(C)$ of finite paths $\pi$ in
the Markov chain $\D(\sigma)$:

\medskip

\noindent
{\em Case (A)}:
$\pi$ is of the form
\begin{multline*}
(s_0, 0)\xrightarrow{\alpha_0} ( (s_0,c_0,s_1), 0)\xrightarrow{\frac{1}{2}} (s_1, 1) \xrightarrow{\alpha_1} ( (s_1,c_1,s_2), 1) \cdots  \\
((s_{M-1}, c_{M-1}, s_{M}), M-1)\xrightarrow{\frac{1}{2}}
(s_M,M)  \xrightarrow{\alpha_M}  ((s_M, c_M, s_{M+1}), M),
\end{multline*}
where $M < 2^m$,  
$c_M \neq F(\CRR_m(M))$, and $c_N = F(\CRR_m(N))$ for all
$N \in [0,M-1]$.  The $\alpha_N$ are probabilities that result from the strategy $\sigma$.
Let $\alpha = \prod_{N \in [0,M]} \alpha_N$.
The probability for the set of all runs from 
$(s_0, 0)$ that  (i)  start with $\pi$, 
then (ii)  proceed to $(q_{F\wedge\neg G},M)$ (if $c_M =1$) 
or to $(q_{\neg F\wedge\neg G},M)$ (if $c_M =0$), and (iii)  do {\em not} visit $R\times\{0\}$   
is at least 
\begin{eqnarray*} 
\label{Case A}
\alpha \cdot   2^{-M+1} \cdot 2^{-|\neg F\wedge \neg G|}  \geq
\alpha \cdot   2^{-(2^m+|\neg F\wedge \neg G|)} .
\end{eqnarray*}
{\em Case (B)}:
$\pi$ is of the form
\begin{multline*}
(s_0, 0)\xrightarrow{\beta_0} ( (s_0,c_0,s_1), 0)\xrightarrow{\frac{1}{2}} (s_1, 1) \xrightarrow{\beta_1} ( (s_1,c_1,s_2), 1) \xrightarrow{\frac{1}{2}}  (s_2,2) \cdots  \\
(s_{M-1}, M-1) \xrightarrow{\beta_{M-1}}
((s_{M-1}, c_{M-1}, s_{M}), M-1) \xrightarrow{\frac{1}{2}} (s_M,M) \xrightarrow{1} (q_G,1),  
\end{multline*}
where $M < 2^m$,  $s_M \in S_f$,  
and $c_N = F(\CRR_m(N))$ for all
$N \in [0,M-1]$.  Let $\beta = \prod_{N \in [0,M-1]} \beta_N$.
The probability for the set of all runs from 
$(s_0, 0)$ that  (i)  start with $\pi$ and (ii)  do {\em not} visit $R\times\{0\}$   
is at least 
\begin{eqnarray*} 
\label{Case B}
\beta  \cdot 2^{-M} \cdot 2^{-|G|}  \geq \beta  \cdot 2^{-(2^m+|\neg F\wedge \neg G|)} .
\end{eqnarray*}
{\em Case (C)}:
$\pi$ is of the form
\begin{multline*}
(s_0, 0)\xrightarrow{\gamma_0} ( (s_0,c_0,s_1), 0)\xrightarrow{\frac{1}{2}} (s_1, 1) \xrightarrow{\gamma_1} ( (s_1,c_1,s_2), 1)  \cdots  \\
(s_{2^m-1},2^m-1) \xrightarrow{\gamma_{2^m-1}}  ((s_{2^m-1}, c_{2^m-1}, s_{2^m}), 2^m-1)\xrightarrow{\frac{1}{2}}
(s_{2^m},2^m), 
\end{multline*}
where $s_{2^m} \not\in S_f$ and $c_N = F(\CRR_m(N))$ for all
$N \in [0,2^m-1]$. Let $\gamma = \prod_{N \in [0,2^m-1]} \gamma_N$.  
The probability of the set of runs in $\D(\sigma)$ that (i) begin with $\pi$, 
then (ii)  proceed (via an intermediate control location of the form $(s_{2^m},b,t)$)
to either $(q_{F\wedge \neg G},2^m)$ or $(q_{\neg F\wedge\neg G},2^m)$ 
and (iii) that do {\em not} reach $R\times\{0\}$ is at least
\begin{eqnarray*}
\label{Case C}
\gamma  \cdot 2^{-(2^m+1)} \cdot 2^{-|\neg F\wedge \neg G|}  = \gamma  \cdot 2^{-(2^m+1+|\neg F\wedge \neg G|)}.
\end{eqnarray*}
Now, the crucial point is that the sum of all values $\alpha$ from (A), all values $\beta$ from (B), and all values
$\gamma$ from (C) is $1$. To see this, note that the nondeterministic choices in $\D$ correspond exactly to 
the selection of transitions in the NFA $A$. But, since $x \not\in L$, every sequence of consecutive transitions in $A$ either (i) 
reads in the $(M+1)^{\text{th}}$ step (for some $0 \leq M \leq 2^m-1$) a symbol different from $F(\CRR_m(M))$
(Case (A)) or (ii) reaches a final state after less than $2^m$ steps (Case (B)), or (iii) make at least
$2^m$ steps and is not in a final state after exactly $2^m$ steps (Case (C)).
Since moreover the set of paths in (A), (B), and (C) are pairwise disjoint,
it follows that the probability of the set of runs that do {\em not} reach $R\times\{0\}$ is at least
$2^{-(2^m+1+|\neg F\wedge \neg G|)}$.
This concludes the proof of the theorem.
\qed
\end{proof}

\bibliographystyle{abbrv}

\def\cprime{$'$}

\section*{Appendix}

Let $M$ be a nondeterministic Turing machine
with a linear ordering on the set of all transition tuples.
Assume furthermore that $M$ does not contain infinite
computation paths. Then, for every input 
$x$, the computation tree $T(x)$ of the machine $M$
on input $x$ is a finite ordered tree. Let
$v_1, v_2, \ldots, v_n$ be a list of all leafs of $T(x)$ in left-to-right
enumeration.  Then the {\em leaf string} 
$\leaf(M,x)$ is the string $a_1 a_2 \cdots a_n$, where
$a_i = 1$ (resp. $a_i = 0$) if $v_i$ is an accepting 
(resp. rejecting) configuration.

\begin{theorem}
Let $A$ be a language in $\PSPACE$. Then
$A$ is $\AC^0$-serializable, i.e.,
there exists a regular language 
$L \subseteq \{0,1\}^*$, a polynomial $p(n)$, and a logspace-uniform
$\AC^0$-circuit family $(B_n)_{n \geq 0}$, where
$B_n$ has exactly $n+p(n)$ many inputs and one output, such that for every 
$x \in \{0,1\}^n$ we have:
$$
x \in A \ \Longleftrightarrow \ B_n(x,0^{p(n)}) \cdots B_n(x,1^{p(n)}) \in L,
$$
where ``$\cdots$'' refers to the lexicographic order on $\{0,1\}^{p(n)}$.
\end{theorem}

\begin{proof}
Let $A \subseteq \{0,1\}^*$ be a language in $\PSPACE$. By the work of \cite{HLSVW93} there exists
a  nondeterministic polynomial time Turing machine 
$$
M = (Q, \Gamma, \Delta, q_0,q_f, \Box)
$$ 
and a regular language 
$K \subseteq \{0,1\}^*$ such that 
\begin{equation} \label{leaf equality}
x \in A  \quad \Longleftrightarrow \quad \leaf(M,x) \in K.
\end{equation}
Here, $Q$ is the set of states, $\Gamma$ is the tape alphabet,
$\Delta \subseteq Q \times \Gamma \times Q \times \Gamma \cup \{L,R \}$
is the set of transition tuples, $q_0$ is the initial state,  $q_f$ is the final (accepting) state, and
$\Box$ is the blank symbol. W.l.o.g. we can 
assume that every computation path of $M$ on an input
of length $n$ has length $q(n)$ for a polynomial $q$.
This can be enforced by introducing a counter.
Note that the counter can be incremented deterministically, hence
the produced leaf string does not change. 
Assume that $\Delta = \{ \delta_1, \ldots, \delta_m \}$, where
$\delta_1 < \delta_1 < \cdots < \delta_m$ is the fixed order
on the transition tuples of $M$.  

Let $\Omega = Q \cup \Gamma \cup \Delta$, where all three sets are assumed to 
be pairwise disjoint. We will encode a computation 
of $M$ of length $q(n)$, starting on input $x \in \Sigma^n$, by a word from
the language
\begin{eqnarray*}
C(x) & = & \{ c_0 t_1 c_1 t_2 \cdots c_{q(n)-1} t_{q(n)} c_{q(n)} \mid t_1, \ldots, t_{q(n)} \in \Delta \\
& &\qquad  c_0 = q_0 x \Box^{q(n)-n}, \
             c_1, \ldots, c_{q(n)} \in \Gamma^* Q \Gamma^+, \\
& &\qquad  |c_1| = \cdots = |c_{q(n)}| = q(n)+1, \ \forall 0 \leq i <  q(n) : c_i \vdash_{t_{i+1}} c_{i+1} \} .
\end{eqnarray*}
Here, $c_i \vdash_{t_{i+1}} c_{i+1}$ means that configuration $c_{i+1}$ results
from configuration $c_i$ by applying transition $t_{i+1}$.
Let $D(x)$ be the subset of $C(x)$ consisting
of all successful computations 
$c_0 t_1 c_1 t_2 \cdots c_{q(n)-1} t_{q(n)} c_{q(n)} \in C(x)$,
where in addition  $c_{q(n)} \in \Gamma^* q_f \Gamma^+$.

Note that every word in $C(x)$ has length $(q(n)+1)^2 + q(n)$.
We use some block encoding $\gamma : \Omega \to \{0,1\}^k$ such that 
$\gamma(\delta_{i+1})$ is lexicographically larger than  $\gamma(\delta_i)$
for $i \in [m-1]$. 
This ensures that if we list all 
bit strings of length $k \cdot ((q(n)+1)^2 + q(n))$ in lexicographic 
order than the subset $C(x)$ of all (encodings of) valid computations appears as a subsequence in the 
same order as in the computation tree $T(x)$.

Let us next describe a logspace-uniform $\AC^0$-circuit family $(C_n)_{n \geq 0}$, where the $n$-th circuit
$C_n$ has $n + k \cdot ((q(n)+1)^2 + q(n))$ many inputs and accepts exactly all strings
of the form $x w$, where $x \in \{0,1\}^n$ and $w \in C(x)$.  
Constructing $C_n$ is tedious but straightforward. 
The most difficult part is to check $c_i \vdash_{t_{i+1}} c_{i+1}$ for all
$0 \leq i <  q(n)$.  For this, we use an AND-gate $g$ with $q(n)$ many children
$g_0, \ldots, g_{q(n)-1}$. Gate $g_i$ is an OR-gate with $q(n)$ many children
$g_{i,1}, \ldots, g_{i,q(n)}$. Gate $g_{i,j}$ evaluates to $1$ if and only if 
$c_{i+1}$ results from $c_i$ by applying the transition $t_{i+1}$ at position
$j$.  To achieve this, $g_{i,j}$ becomes an AND-gate with $k(q(n)+1)$ many
input gates. Each of these gates compares two corresponding bits in the 
$\gamma$-encodings of $c_i$ and $c_{i+1}$.  It should be clear that
such a circuit $C_n$ can be built in logarithmic space.
Analogously we can construct a
logspace-uniform $\AC^0$-circuit family $(D_n)_{n \geq 0}$
which accepts all strings of the form $x w$, where 
$x \in \{0,1\}^n$ and $w \in D(x)$.

Finally, we construct from the two families 
$(C_n)_{n \geq 0}$ and $(D_n)_{n \geq 0}$ a new 
logspace-uniform $\AC^0$-circuit family $(B_n)_{n \geq 0}$, where
$B_n$ has $n + k \cdot ((q(n)+1)^2 + q(n)) + 1$ many inputs.
On input $x w 0$ (with $x \in \Sigma^n$) it outputs 
$C_n(xw)$. On input $x w 1$, $B_n$ outputs 
$D_n(xw)$. Now, let us construct from the regular language $K \subseteq \{0,1\}^*$ 
the new regular language 
$L = \varphi( K \mid\!\mid \{ a \}^* )$,
where $\mid\!\mid$ is the shuffle operator,
$a \not\in \{0,1\}$ is a new symbol, and $\varphi$ is the homomorphism with 
$\varphi(a) = 00$, $\varphi(0) = 10$, $\varphi(1) = 11$. 

The regular language $L$, the polynomial $p(n) = k \cdot ((q(n)+1)^2 + q(n)) + 1$, and 
the circuit family $(B_n)_{n \geq 0}$ fulfill the requirements from the theorem. 
\qed
\end{proof}

\end{document}